%% file: rank_min_v2.tex
\documentclass[10pt, letterpaper, twocolumn]{IEEEtran}

\input{preamble}

\title{Rank Minimization over Finite Fields: Fundamental Limits and Coding-Theoretic Interpretations}
\author{Vincent Y.~F.\ Tan,   Laura Balzano, \IEEEmembership{Student Member, IEEE},    Stark C. Draper, \IEEEmembership{Member, IEEE}\thanks{This work is supported in part by the Air Force Office of Scientific Research under grant FA9550-09-1-0140 and by the National Science Foundation under grant CCF 0963834.   V.~Y.~F.\ Tan is also supported by  A*STAR Singapore.    This paper was  presented in part at the IEEE International Symposium on Information Theory (ISIT), St.\ Petersburg, Russia, August 2011~\cite{Tan11}. } \thanks{The authors are with the Department of Electrical and Computer Engineering (ECE), University of Wisconsin, Madison, WI, 53706, USA  (emails: {vtan@wisc.edu}; {sunbeam@ece.wisc.edu}; {sdraper@ece.wisc.edu}). The first author is also affiliated  to the Laboratory for Information and Decision Systems (LIDS), Massachusetts Institute of Technology (MIT), Cambridge, MA, 02139, USA (email: {vtan@mit.edu}).  } \thanks{Copyright (c) 2011 IEEE. Personal use of this material is permitted.  However, permission to use this material for any other purposes must be obtained from the IEEE by sending a request to pubs-permissions@ieee.org.}}

\begin{document}
\maketitle

\begin{abstract}
This paper establishes  information-theoretic limits for estimating a finite field low-rank matrix  given random linear measurements of it.  These linear measurements are obtained by taking inner products of the  low-rank   matrix with  random sensing matrices. Necessary and sufficient conditions on the number of measurements required are provided. It is shown that these conditions are sharp and the minimum-rank decoder is asymptotically optimal. The reliability function of   this decoder   is also derived by appealing to de Caen's lower bound on the probability of a union. The   sufficient condition  also holds when the sensing matrices are sparse -- a scenario that may be amenable to efficient decoding. More precisely, it is shown that if the $n\times n$-sensing matrices   contain, on average, $\Omega( {n}{\log n})$ entries, the number of measurements required is the same as that when the sensing matrices are   dense and contain entries drawn uniformly at random from the field. Analogies are drawn between the above results and rank-metric codes in the coding theory literature. In fact, we are also strongly motivated by  understanding when  minimum rank distance decoding of random rank-metric codes succeeds. To this end, we derive  minimum  distance properties of equiprobable and sparse rank-metric codes. These distance properties provide a precise geometric interpretation of the fact that the sparse ensemble requires as few measurements as the dense one.
\end{abstract}


\section{Introduction}
This paper considers the problem of rank minimization over finite fields. Our work attempts to connect two seemingly disparate areas of study that have, by themselves, become    popular in the information theory   community in recent years: (i) the theory of matrix completion~\cite{Can10,Can09, Rec09} and rank minimization~\cite{Rec09a, Meka} over the reals and (ii) rank-metric codes~\cite{Gab85, Roth, Loi06, Sil08, Mon07, Gad08}, which are the rank distance analogs of binary   block codes endowed with  the Hamming metric. The work herein provides a starting point for investigating the potential impact of the low-rank assumption on information and coding theory. We   provide a brief review of these two areas of study. 

The problem of matrix completion~\cite{Can10,Can09, Rec09} can be stated as follows: One is given a  subset of  noiseless or noisy  entries of a low-rank matrix (with entries over the reals),  and  is then required to estimate all the remaining entries. This problem has a variety of applications from collaborative filtering (e.g., Netflix prize~\cite{NetflixPrize}) to obtaining the minimal realization of a linear dynamical system~\cite{Faz01}. Algorithms based on the nuclear norm (sum of singular values) convex  relaxation of the rank function~\cite{Faz01, Faz03} have enjoyed tremendous successes.  A generalization of the  matrix completion  problem is the   rank minimization problem~\cite{Rec09a, Meka} where, instead of being given entries of the low-rank matrix, one is given arbitrary linear measurements of it. These linear measurements are obtained by taking inner products of the  unknown matrix with  sensing matrices.   The nuclear norm heuristic has also been shown to be extremely effective in estimating the unknown  low-rank  matrix. Theoretical results~\cite{Rec09a, Meka} are typically of the following flavour: If the number of measurements (also known as the measurement complexity) exceeds a small multiple of the product of the dimension of the matrix and its rank, then optimizing the nuclear-norm heuristic yields the same  (optimal) solution as the   rank minimization problem under certain conditions on the sensing matrices. Note that in the case of real matrices, if the observations (or the entries) are noisy, perfect reconstruction is impossible. As we shall see in Section~\ref{sec:noisy}, this is not the case in the finite field setting. We can recover the underlying matrix {\em exactly} albeit at the cost of a higher measurement complexity.

Rank-metric codes~\cite{Gab85, Roth, Loi06, Sil08, Mon07, Gad08}   are  subsets of finite field matrices endowed with the rank-metric. We will be concerned with {\em linear} rank-metric codes,  which may be characterized by a family of parity-check matrices, which  are equivalent to the sensing matrices in the rank minimization problem. 


\subsection{Motivations} \label{sec:motivation}
Besides  analyzing the measurement complexity for rank minimization over finite fields, this paper is also motivated by two applications in coding. The first is index coding with side information~\cite{Bar11}. In brief, a sender wants to communicate the $l$-th coordinate of a length-$L$ bit string to the $l$-th of $L$ receivers. Furthermore, each of the $L$ receivers knows a subset of the coordinates of the bit string. These subsets can be represented by (the neighbourhoods of) a graph.  Bar-Yossef et al.~\cite{Bar11} showed that the    linear version of this problem  reduces to  a rank minimization problem. In previous works, the graph is deterministic. Our work, and in particular the rank minimization problem considered herein, can be cast as the solution  of a linear  index coding problem with a {\em random} side information graph.  

Second, we are   interested in properties of the rank-metric coding problem~\cite{Sil08}.    Here, we are given a set of matrix-valued codewords that form a linear rank-metric code $\scC$. A codeword  $\bC^*\in\scC$ is transmitted across a noisy finite field   matrix-valued channel which induces an additive {\em error matrix} $\bX$. This error matrix $\bX$ is assumed to be low rank. For example, $\bX$ could be  a matrix   induced by the crisscross error model in data arrays~\cite{Roth97}. In the crisscross error model, $\bX$ is a sparse low rank matrix in which the non-zero elements are restricted to a small number of rows and columns. The received matrix is  $\bR:= \bC^*+\bX.$ The minimum distance decoding  problem is given by the following:
\begin{equation} \label{eqn:rank_decode}
\hat{\bC} := \argmin_{\bC\in \scC} \,\, \rank(\bR-\bC).
\end{equation}

We would like to study when problem~\eqref{eqn:rank_decode} succeeds (i.e., uniquely recovers the true codeword $\bC^*$) with high probability\footnote{Here and in the following, {\em with high probability} means with probability tending to one as the problem size  tends to infinity.} (w.h.p.) given that $\scC$ is a random   code characterized by either dense or sparse random parity-check   matrices and $\bX$ is a deterministic error matrix. But why analyze {\em random} codes?  Our study of random (instead of deterministic) codes is motivated by the fact that    data arrays that arise in applications are often corrupted by crisscross error patterns~\cite{Roth97}. Decoding techniques used in the rank-metric literature such as error trapping~\cite{Mon07, Sil10} are unfortunately not able to correct such error patterns  because  they are highly structured and hence the ``error traps'' would miss (or not be able to correct) a non-trivial subset of errors. Indeed, the success such an error trapping strategy hinges strongly on the assumption that the underlying  low-rank error matrix $\bX$ is drawn {\em uniformly at random} over all matrices whose rank is  $r$~\cite[Sec.~IV]{Sil10} (so subspaces can be trapped).   The decoding technique in \cite{Roth97} is specific to correcting crisscross error patterns. In contrast, in this work, we are able to derive  distance properties of random rank-metric codes and to show that given sufficiently many constraints on the codewords, {\em all} error patterns of rank no greater than $r$ can be successfully corrected.   Although our derivations are similar in spirit to those in Barg and Forney~\cite{Barg02}, our starting point is rather different. In particular, we   combine the use of techniques from~\cite{Gall} and those in~\cite{Barg02}.  

We are also motivated by the fact that error exponent-like results for  matrix-valued finite field channels are,  to the best of the authors' knowledge, not available in the literature. Such channels have been popularized by the seminal work in~\cite{Kot08}. Capacity results for specific channel models such as the uniform given rank (u.g.r.)  multiplicative noise model~\cite{Nob11} have recently been derived. In this work, we derive the error exponent for the minimum-rank decoder $E(R)$ (for the additive noise  model).
This fills an important gap in the literature.

\subsection{Main Contributions}
We summarize our four main contributions in this work. 

Firstly, by using a standard converse technique (Fano's inequality), we derive a necessary condition on the number of measurements required for estimating a low-rank matrix. Furthermore, under the assumption that the linear measurements are obtained by taking inner products of the unknown matrix with   sensing matrices containing independent entries that are equiprobable (in $\bbF_q$), we demonstrate an achievability procedure, called the min-rank decoder, that matches the information-theoretic lower bound on the number of measurements required. Hence, the sufficient condition is sharp.  Extensions to the noisy case are also discussed.  Note that in this paper, we are not as concerned with the computational complexity of recovering the unknown low-rank matrix as compared to the fundamental limits of doing so.

Secondly, we derive the reliability function (error exponent) $E(R)$ of the min-rank decoder by using de Caen's lower bound on the probability of a union~\cite{deCaen}.  The use of de Caen's bound to obtain estimates of the reliability function (or probability of error) is not new. See the works by S\'{e}guin~\cite{Seg98}  and Cohen and Merhav~\cite{Coh04}   for example. However, by exploiting pairwise independence of  constituent error events, we   not only derive upper and lower bounds on $E(R)$, we show that these bounds are, in fact, tight for all rates (for the min-rank decoder). We derive the corresponding error exponents for codes in~\cite{Gab85} and~\cite{Sil10} and make comparisons between the error exponents.

Thirdly, we show that if the fraction of non-zero entries of the sensing or measurement matrices scales (on average) as $\Omega(\frac{\log n}{n})$ (where the matrix is of size $n\times n$), the min-rank decoder achieves the information-theoretic lower bound. Thus, if the average number of entries in each sparse sensing matrix is   $\Omega(n\log n)$ (which is much fewer than $n^2$),  we can show that, very surprisingly, the number of linear measurements required for reliable reconstruction of the unknown low-rank matrix  is  exactly the same as that for the equiprobable (dense) case. This main result  of ours opens the possibility for the development of  efficient, message-passing decoding algorithms based on sparse parity-check matrices~\cite{Bar10}.

Finally, we draw analogies between the above results and rank-metric codes~\cite{Gab85, Roth, Loi06, Sil08, Mon07, Gad08} in the coding theory literature. We derive minimum (rank) distance properties  of the equiprobable random ensemble and the sparse random ensemble. Using elementary techniques, we derive an analog of the Gilbert-Varshamov distance for the random rank-metric code. We also compare and contrast our result to classical binary linear block codes with the Hamming metric~\cite{Barg02}. From our analyses in this section, we obtain geometric intuitions to explain why minimum rank decoding performs well even when the sensing matrices are sparse. We also use these geometric intuitions to guide our derivation of strong recovery guarantees along the lines of the recent work by Eldar et al.~\cite{Eldar11}.

\subsection{Related Work}
There is a wealth of literature on  rank minimization   to which we will   not be able to do justice here. See  for example the seminal works by Fazel et al.~\cite{Faz01, Faz03} and the subsequent works by other authors~\cite{Can10,Can09, Rec09} (and the references therein). However, all these works focus on the case where the unknown matrix is over the reals. We are interested in the finite field setting because such a problem has many connections with and applications to coding and information theory~\cite{Roth97, Pap10, Bar11}. The analogous problem for the reals was considered by Eldar et al.~\cite{Eldar11}. The results in~\cite{Eldar11}, developed for dense sensing matrices with i.i.d.\ Gaussian  entries, mirror those in this paper but only achievability results (sufficient conditions) are provided. We additionally analyze the sparse setting.

Our work is partially inspired by~\cite{Dra09} where fundamental limits for compressed sensing over finite fields were derived. To the best of our knowledge,  Vishwanath's work~\cite{Vis10} is the only one that employs information-theoretic techniques to derive  necessary and sufficient conditions on the number of measurements required for reliable matrix completion (or rank minimization). It was shown using typicality arguments that the number of measurements required is within a logarithmic factor of the lower bound. Our setting is different because we assume that we have linear measurements instead of randomly sampled entries. We are able to show that the achievability and converse match   for a family of random  sensing   matrices.   Emad and Milenkovic~\cite{Emad11}  recently extended the analyses in the conference version~\cite{Tan11}  of this paper to the tensor case,  where the rank, the order of the tensor and the number of measurements grow simultaneously with the size of the matrix. We   compare and contrast our decoder and analysis for the noisy case to that in~\cite{Emad11}. Another recent related work is that by Kakhaki et al.\ \cite{Kak11} where the authors considered   the binary erasure channel (BEC) and binary symmetric channel (BSC)  and empirically studied the  error exponents for codes whose {\em generator matrices} are random and sparse. For the BEC, the authors showed that there exist capacity-achieving codes with generator matrices whose sparsity factor (density)   is   $O(\frac{\log n}{n})$ (similar to this work). However, motivated by the fact that sparse parity-check matrices  may make   decoding   amenable to lower complexity message-passing type decoders, we analyze the scenario where the {\em parity-check matrices} are sparse.  

\begin{table}[t]
\centering
\caption{Comparison of our work (\underline{T}an-\underline{B}alzano-\underline{D}raper) to existing coding-theoretic techniques for rank minimization } 
    \begin{tabular}{ | c|| c  |c |}    \hline 
      Paper  & Code Structure &    Decoding Technique  \\ \hline  
    Gabidulin  \cite{Gab85} & Algebraic   &  Berlekamp-Massey  \\ \hline
     SKK \cite{Sil08} & Algebraic  &   Extended Berlekamp-Massey  \\ \hline
 MU \cite{Mon07} & Factor Graph &    Error Trapping \& Message Passing \\ \hline
  SKK \cite{Sil10} & Error Trapping &    Error Trapping \\ \hline
 GLS \cite{Gro81}  & Perfect Graph &   Semidefinite Program (Ellipsoid)\\\hline 
TBD  &   See Table \ref{tab:comp2}  &    Min-Rank Decoder (Section~\ref{sec:nonexhaust})  \\ 
    \hline    \end{tabular}
    \label{tab:comp}
\end{table}

The family of codes known as rank-metric codes~\cite{Gab85,Roth, Loi06,Sil08, Mon07 , Gad08}, which are the the rank-distance analog of binary   block codes equipped with the Hamming metric,  bears a striking similarity to the rank minimization problem over finite fields. Comparisons between this work and related works in the coding theory literature are summarized in Table~\ref{tab:comp}. Our contributions in the various sections of this paper, and other pertinent references, are   summarized   in Table~\ref{tab:comp2}. We will further elaborate on these comparisons   in Section~\ref{sec:coding_conn}.

\subsection{Outline of Paper}
Section \ref{sec:model1} details our notational choices, describes the measurement models  and states the problem. In Section~\ref{sec:nec}, we use Fano's inequality to derive a lower bound on the number of measurements   for reconstructing the unknown low-rank matrix. In Section~\ref{sec:uar}, we consider the  uniformly at random (or equiprobable) model where the entries of the measurement matrices are selected  independently and uniformly at random from $\bbF_q$. We derive a sufficient condition for reliable recovery and the  reliability function of the min-rank decoder using de Caen's lower bound. The results are then  extended to the noisy scenario in Section~\ref{sec:noisy}. Section~\ref{sec:sparse}, which contains our main result, considers the case where the measurement matrices are sparse. We derive a sufficient condition on the sparsity factor (density) as well as the number of measurements for reliable recovery. Section~\ref{sec:dist} is devoted to understanding and interpreting the above results from a coding-theoretic perspective.  In Section~\ref{sec:nonexhaust}, we  provide a   procedure to search for the low-rank matrix by exploiting indeterminacies  in the problem.  Discussions and conclusions are provided in Section~\ref{sec:concl}. The lengthier  proofs are deferred to the appendices. 
\begin{table}[t]
\centering
\caption{Comparisons between the  results in various sections of this paper and other related works  } 
\begin{tabular}{ | c|| c  |c |}    \hline 
 Parity-check & Random  & Deterministic   \\
    matrix  $\rvbH_a$  &  low-rank      matrix $\rvbX$ &  low-rank     matrix $\bX$   \\ \hline  
Random, dense &   Section~\ref{sec:uar}   &  Section~\ref{sec:uar} \\ \hline
Deterministic, dense &  Section~\ref{sec:uar},  \cite{Sil10}   &    $\!\!$   Section \ref{sec:str_ach}, \cite{Gab85}, \cite{Sil08}  $\!\!\!$  \\ \hline
Random, sparse & Section~\ref{sec:sparse}  &    Section~\ref{sec:sparse} \\ \hline
Deterministic, sparse &  $\!\!$ Section~\ref{sec:sparse}, \cite{Mon07,Sil10} $\!\!$ & Section \ref{sec:str_ach}    \\ \hline
\end{tabular}
\label{tab:comp2}
\end{table}

\section{Problem Setup and Model} \label{sec:model1}
In this section, we state our notational conventions, describe the system model and state the problem. We   also  distinguish between the  two related notions of weak and strong recovery. 

\subsection{Notation}
In this paper we adopt the following set of notations: Serif font and san-serif font    denote deterministic and random quantities respectively. Bold-face upper-case and  bold-face lower-case denote matrices  and (column) vectors respectively.  Thus,  $y$, $\rvy$, $\bX$ and $\rvbX$ denote  a deterministic scalar, a scalar-valued random variable, a  deterministic matrix and a random matrix respectively. Random functions will also be denoted in san-serif font. Sets (and events) are denoted with calligraphic font (e.g., $\calU$ or $\scC$). The cardinality of a finite set $\calU$ is denoted as $|\calU|$. For a prime power  $q$, we denote the finite  (Galois) field with $q$ elements  as $\bbF_q$. If $q$ is prime, one can identify $\bbF_q$ with $\bbZ_q =\{0,\ldots, q-1\}$, the set of the integers modulo $q$. The set of $m\times n$ matrices with entries in $\bbF_q$ is denoted as $\bbF_q^{m\times n}$. For simplicity, we let $[k]:=\{1,\ldots, k\}$ and $\rvy^k:=(\rvy_1,\ldots, \rvy_k)$. For a matrix $\bM$, the notations $\|\bM\|_0$ and $\rank(\bM)$  respectively denote  the number of non-zero elements in $\bM$ (the Hamming weight) and the rank of $\bM$ in $\bbF_q$. For a matrix $\bM\in\bbF_q^{m\times n}$, we also use the notation $\vect(\bM) \in\bbF_q^{mn}$ to denote vectorization of  $\bM$ with its columns  stacked on top of one another.   For a real number $b$,  the notation $|b|^+$ is defined as $\max\{b, 0\}$.   Asymptotic notation such as $O(\fndot),\Omega(\fndot)$ and   $o(\fndot)$ will be used throughout. See~\cite[Sec.~I.3]{Cor03} for definitions. For the reader's convenience, we have summarized the   symbols used in   this paper in Table~\ref{tab:not}. 

\begin{table}
\centering
\caption{Table of      symbols used in this paper   } 
\begin{tabular}{ | c|| c  |c |}    \hline 
Notation & Definition & Section \\\hline
$k $ & Number of measurements & Section \ref{sec:model}\\ \hline
$r/n\to\gamma$ & Rank-dimension ratio & Section \ref{sec:model}\\ \hline
$\sigma=\|\bw\|_0 / n^2$ &  Deterministic noise parameter &  Section \ref{sec:det} \\ \hline  
$\alpha= k/n^2$ &  Measurement scaling parameter &  Section \ref{sec:rand} \\ \hline  
$p = \bbE \|\rvbw\|_0 /k$ &  Random noise parameter &  Section \ref{sec:rand} \\ \hline  
$\delta \!=\!\bbE \|\rvbH_a\|_0 /n^2\!\!$ &  Sparsity factor &  Section \ref{sec:sparse} \\ \hline  
$\rvN_{\scC}(r)$ & $\!\!\!\!$  Num.\ of matrices of rank $r$ in $\scC$ $\!\!\!\!\!\!$ &  Section \ref{sec:dist} \\ \hline  
$\rvd(\scC)$ & Minimum rank distance of $\scC$   &  Section \ref{sec:dist} \\ \hline  
\end{tabular}
\label{tab:not}
\end{table}

\subsection{System Model} \label{sec:model}
We are interested in the following  model: Let $\bX$ be an unknown (deterministic or random) square\footnote{Our results are not restricted to the case where $\bX$ is square but for the most part in this paper, we assume that  $\bX$ is square for ease of exposition.} matrix in $\bbF_q^{n\times n}$ whose rank  is less than or equal to $r$, i.e., $\rank(\bX)\le r$. The upper bound on the rank $r$ is allowed to be a function of $n$, i.e., $r=r_n$. We assume that $r/n\to \gamma$ and we say that the limit $\gamma\in [0,1]$ is  the {\em rank-dimension ratio}.\footnote{Our results   also include the regime where $r  =o(n)$ but the case where $r=\Theta(n)$ (and $\gamma$ is the proportionality constant) is of   greater interest and significance. This is because the rank $r$ grows as rapidly as possible and hence this regime is the most challenging.  Note that if $r/n\to \gamma=1$, then we would need $n^2$ measurements to recover $\bX$ since we are not making any low rank assumptions on it. This is corroborated by the converse  in Proposition~\ref{prop:converse}.  }  We would like to recover or estimate $\bX$ from $k$ {\em linear measurements} 
\begin{equation}
\rvy_a =\lrangle{\rvbH_a}{\bX} : =\sum_{(i,j)\in [n]^2} [\rvbH_a]_{i,j} [\bX]_{i,j} \qquad a\in [k], \label{eqn:linear_meas}
\end{equation}
i.e.,  $\rvy_a$ is the  trace of $\rvbH_a \bX^T$. In~\eqref{eqn:linear_meas}, the {\em sensing} or {\em measurement} matrices $\rvbH_a\in \bbF_q^{n\times n},a\in [k]$, are random matrices chosen according to some probability mass function (pmf). The $k$ scalar {\em measurements} $\rvy_a \in \bbF_q, a\in [k]$, are available for estimating $\bX$.    We will  operate in the so-called high-dimensional setting and allow the number of measurements  $k$ to depend on $n$, i.e., $k=k_n$.  Multiplication and addition  in~\eqref{eqn:linear_meas} are performed in $\bbF_q$. In the subsequent sections, we will also be interested in a generalization of the model in~\eqref{eqn:linear_meas} where the measurements $\rvy_a,a\in [k]$, may not be noiseless, i.e., 
\begin{equation}
\rvy_a =\lrangle{\rvbH_a}{\bX} + \rvw_a,\qquad a\in [k]\label{eqn:noisy_meas},
\end{equation}
where $\rvw_a,a\in [k]$, represents  random or deterministic  noise. We will specify precise noise models in Section~\ref{sec:noisy}. 

The measurement models we are concerned with in this paper, \eqref{eqn:linear_meas} and~\eqref{eqn:noisy_meas}, are somewhat different from the   {\em  matrix completion problem}~\cite{Can10,Can09, Rec09}. In the matrix completion setup, a subset of entries $\Omega\subset [n]^2$ in the matrix $\bX$ is  observed and one would like to ``fill in'' the rest of the entries  assuming the matrix is low-rank. This model can be captured by \eqref{eqn:linear_meas} by choosing each  sensing matrix $\rvbH_a$ to be   non-zero only in a single position.  Assuming $\rvbH_a\ne \rvbH_{a'}$ for all $a\ne a'$, the number of measurements is $k = |\Omega|$. In contrast,  our measurement models in~\eqref{eqn:linear_meas} and~\eqref{eqn:noisy_meas}   do not assume that $\|\rvbH_a\|_0=1$.  The sensing matrices are, in general, dense although in Section~\ref{sec:sparse}, we also analyze the scenario where $\rvbH_a$ is relatively sparse.  Our setting is more similar in spirit  to the rank minimization problems analyzed in Recht et al.~\cite{Rec09a}, Meka et al.~\cite{Meka} and Eldar et al.~\cite{Eldar11}. However, these   works focus on problems in the reals whereas our focus is  the finite field setting. 

\subsection{Problem Statement}
Our objective is to estimate the  unknown low-rank matrix $\bX$ given $\rvy^k$ (and the measurement matrices $\rvbH_a,a\in [k]$). In general, given the  measurement   model in~\eqref{eqn:linear_meas} and  without any  assumptions on $\bX$, the problem is ill-posed and it is not possible to recover $\bX$  if $k<n^2$. However, because $\bX$ is assumed to have rank no larger than $r$ (and $r/n\to\gamma$), we can exploit this additional information to estimate $\bX$ with  $k<n^2$ measurements.   Our goal in this paper is to characterize necessary and sufficient conditions  on the number of measurements $k$  as $n$ becomes large assuming a particular pmf governing the sensing matrices $\rvbH_a, a\in [k]$ and under various (random and deterministic) models on $\rvbX$.

\subsection{Weak Versus Strong Recovery} \label{sec:weakvsstr}
In this paper, we will focus (in Sections~\ref{sec:nec} to \ref{sec:sparse})   on the so-called {\em weak recovery} problem where the unknown low-rank matrix $\bX$  is {\em fixed} and we ask how many measurements $k$ are sufficient to recover $\bX$  (and what the procedure is for doing so).   However, there is also a companion problem known as the {\em strong recovery} problem, where one would like to recover {\em all} matrices in $\bbF_q^{n\times n}$ with rank no larger than $r$.  A familiar version of this distinction also arises in compressed sensing.\footnote{Analogously in compressed sensing, consider the combinatorial $\ell_0$-norm optimization  problem  $\min_{\tilde{\bx}\in\bbF^n} \{\|\tilde{\bx} \|_0:\rvbA\tilde{\bx}=\rvby\}$, where the field $\bbF$ can either be the reals $\bbR$~\cite{Eldar11} or a finite field $\bbF_q$~\cite{Dra09}. It can be seen that if we want to recover {\em fixed but unknown} $s$-sparse vector $\bx$ (weak recovery),    $s + 1$ linear measurements suffice w.h.p.  However, for strong recovery where we would like to guarantee recovery {\em for all} $s$-sparse vectors, we need to ensure that the nullspace of the measurement matrix  $\rvbA$ is disjoint from the set of $2s$-sparse vectors. Thus, w.h.p.,  $2s$   measurements are required for strong recovery~\cite{Eldar11,Dra09}.  }

More precisely, given $k$  sensing matrices $\rvbH_a, a\in[k]$,    we define the linear operator $\rvbH :\bbF_q^{n\times n}\to\bbF_q^k$ as 
\begin{equation}
 \rvbH(\bX) := [ \langle \rvbH_1, \bX \rangle, \langle \rvbH_2, \bX \rangle, \ldots,  \langle \rvbH_k, \bX \rangle ]^T. \label{eqn:linear_stack}
\end{equation}
Then, a necessary and sufficient condition  for strong recovery is that the   operator $\rvbH$  is injective when restricted to the set of all matrices of rank-$2r$ (or less). In other words,  there are no rank-$2r$ (or less) matrices  in the nullspace of the operator $\rvbH$ \cite[Sec.~2]{Eldar11}. This can be observed by noting that for two matrices $\bX_1$ and $\bX_2$ of rank-$r$  (or less) that generate the same linear observations (i.e., $\rvbH(\bX_1)=\rvbH(\bX_2)$), their difference $\bX_1-\bX_2$ has rank at most $2r$ by the triangle inequality.\footnote{Note that $(\bA,\bB)\mapsto\rank(\bA-\bB)$ is a metric on the space of matrices.} We would thus like to find conditions  on $k$ (via, for example, the geometry of the random code) such  that  the following subset of $\bbF_q^{n\times n}$
\begin{equation}
\calR_{2r}^{(n)}:=\{\bX\in \bbF_q^{n\times n}: \rank(\bX)\le 2r\}
\end{equation}
is disjoint from the nullspace of $\rvbH$  with   probability tending to one as $n$ grows. As mentioned in Section~\ref{sec:model}, we allow $r$ to grow linearly with $n$ (with proportionality constant $\gamma$). Under the  condition that $\calR_{2r}^{(n)}\cap \, \mathrm{nullspace}(\rvbH)=\emptyset$, the solution to the rank minimization problem [stated precisely in \eqref{eqn:minrank} below] is  unique and correct {\em for all} low-rank matrices with probability tending to one as $n$ grows. As we shall see in Section~\ref{sec:str_ach}, the conditions on $k$ for strong recovery are  more stringent than those for weak recovery.   See the recent paper by Eldar et al.~\cite[Sec.~2]{Eldar11} for further discussions on weak versus strong recovery in the real field setting.  

\subsection{Bounds on  the number of  low-rank matrices} \label{sec:bounds_number}
In the sequel, we will find it useful to leverage the following lemma, which is a combination of results stated in~\cite[Lemma~4]{Kot08}, \cite[Proposition~1]{Loi06} and~\cite[Lemma~5]{Gad08}.
\begin{lemma}[Bounds on the number of  low-rank matrices]  \label{lem:num_mat} Let $\Phi_q(n,r)$ and  $\Psi_q(n,r)$  respectively be the number of matrices in $\bbF_q^{n\times n}$ of rank exactly $r$ and the number of matrices in $\bbF_q^{n\times n}$ of rank less than or equal to $r$. Note that $\Psi_q(n,r) = \sum_{l=0}^r \Phi_q(n,l)$.  The following bounds hold:
\begin{align}
q^{(2n-2)r-r^2} &\le\Phi_q(n,r)\le 4 q^{2 nr-r^2 } ,\label{eqn:exact_r}   \\
q^{2nr-r^2} &\le\Psi_q(n,r)\le 4 q^{2 nr-r^2 } . \label{eqn:less_r} 
\end{align}
\end{lemma}
In other words,   we have from~\eqref{eqn:less_r}  and the fact that $r/n\to\gamma$ that  $|\frac{1}{n^2}\log_q \Psi_q(n,r)-2\gamma (1-\gamma/2) |\to 0$.

\section{A Necessary Condition  for Recovery} \label{sec:nec}
This section presents a necessary condition  on the scaling of $k$ with $n$ for the matrix $\bX$ to be recovered {\em reliably}, i.e., for the error probability in estimating $\bX$ to tend to zero as $n$ grows. As with  most other converse statements in information theory, it is necessary  to assume a statistical model on the unknown object, in this case $\bX$. Hence, in this section, we denote the unknown low-rank matrix  as $\rvbX$ (a random variable). We  also assume that $\rvbX$ is drawn {\em uniformly at random} from the set of matrices in $\bbF_q^{n\times n}$ of rank less than or equal to $r$.  For an {\em estimator}  (deterministic or random function) $\hat{\rvbX}:\bbF_q^k \times (\bbF_q^{n\times n})^k\to \bbF_q^{n\times n}$ whose  range is the set of all $\bbF_q^{n\times n}$-matrices whose rank is less than or equal to $r$, we define the error event:
\begin{equation}
\tilde{\calE}_n:=  \{ \hat{\rvbX}(\rvy^k,\rvbH^k)  \ne \rvbX \}. \label{eqn:tildeE}
\end{equation}
This is the event that the estimate $\hat{\rvbX}(\rvy^k,\rvbH^k) $ is not equal to the true low-rank matrix $\rvbX$.  We emphasize that the estimator   can either be deterministic or random. In addition,  the arguments $(\rvy^k,\rvbH^k)$  are random so $\hat{\rvbX}(\rvy^k,\rvbH^k)$ in the definition of $\tilde{\calE}_n$ is a random matrix.  We can demonstrate the following:
\begin{proposition}[Converse] \label{prop:converse}
Fix $\veps > 0$ and assume that $\rvbX$ is drawn uniformly at random from all  matrices of rank less than or equal to $r$. Also, assume $\rvbX$ is independent of $\rvbH^k$. If, 
\begin{equation}
k<(2-\veps)\gamma \left(1-\gamma/2\right)n^2\label{eqn:conversek}
\end{equation}
then for any estimator $\hat{\rvbX}$ whose range is the set of $\bbF_q^{n\times n}$-matrices  whose rank is less than or equal to $r$, $\bbP(\tilde{\calE}_n)\ge {\veps}/{4} > 0$ for all $n$ sufficiently large. 
\end{proposition}
Proposition~\ref{prop:converse} states that the number of measurements $k$ must   exceed $2nr-r^2$ (which is approximately $2\gamma(1-\gamma/2)n^2$) for recovery of $\rvbX$ to be {\em reliable}, i.e., for the probability of $\tilde{\calE}_n$ to tend to zero as $n$ grows. From a linear algebraic perspective, this means we need at least as many measurements as there are degrees of freedom in the unknown object $\rvbX$. Clearly, the bound  in~\eqref{eqn:conversek} applies to both the noisy and the noiseless models introduced in Section~\ref{sec:model}. The proof involves an elementary application of Fano's inequality \cite[Sec.~2.10]{Cov06}. 
\begin{proof}
Consider the following lower bounds on the probability of error $\bbP(\tilde{\calE}_n)$:
\begin{align}
&\bbP(\hat{\rvbX} \ne \rvbX)  \gea  \frac{H(\rvbX|\rvy^k,\rvbH^k)\!-\!1}{\log_q \Psi_q(n,r)}\!=\! \frac{H(\rvbX)\!-\!I(\rvbX;\rvy^k,\rvbH^k)\!-\! 1}{\log_q \Psi_q(n,r)}  \nn\\
 &\eqb\frac{H(\rvbX)-I(\rvbX;\rvy^k|\rvbH^k)-1}{\log_q \Psi_q(n,r)} =\frac{H(\rvbX)-H(\rvy^k | \rvbH^k)-1}{\log_q \Psi_q(n,r)}  \nn\\
&\gec \frac{H(\rvbX) -k      -1}{\log_q \Psi_q(n,r)}   \eqd 1-\frac{k}{\log_q \Psi_q(n,r)}-o(1), \label{eqn:end_fanos}
\end{align}
where  $(a)$ is by Fano's inequality (estimating $\rvbX$ given $\rvy^k$ and $\rvbH^k$),   $(b)$ is because $\rvbH^k$ is independent of $\rvbX$ so  $I(\rvbX;\rvy^k,\rvbH^k)=I(\rvbX;\rvy^k|\rvbH^k)+I(\rvbX;\rvbH^k)= I(\rvbX;\rvy^k|\rvbH^k)$. Inequality    $(c)$  is  due to the fact that $\rvy_a$ is $q$-ary for all $a\in [k]$  so 
\begin{equation}
 H(\rvy^k|\rvbH^k) \le H(\rvy^k) \le  kH(\rvy_1) \le k \log_q q=  k,  \label{eqn:equivo}
\end{equation}
and finally,  $(d)$ is due to the uniformity of $\rvbX$. It can be easily verified that if $k$ satisfies~\eqref{eqn:conversek} for some $\veps>0$, then
$
{k}/{\log_q \Psi_q(n,r)}\le   1  -\veps/3
$ for   $n$ sufficiently large by the lower bound in~\eqref{eqn:less_r} and the convergence $r/n\to \gamma$.  Hence, \eqref{eqn:end_fanos} is larger than $\veps/4$ for all $n$ sufficiently large. 
\end{proof}

We emphasize that the assumption that the sensing matrices $\rvbH_a,a\in [k]$ are statistically independent of the unknown low-rank matrix $\rvbX$ is important. This is  to ensure the validity of equality $(b)$ in~\eqref{eqn:end_fanos}.  This assumption is not a restrictive one in practice since the sensing mechanism is usually independent of the unknown matrix.
%
%

\section{Uniformly Random Sensing Matrices: The Noiseless Case} \label{sec:uar}
In this section, we assume the noiseless linear model in~\eqref{eqn:linear_meas} and provide sufficient conditions for the recovery of  a {\em fixed} $\bX$ (a deterministic low-rank matrix) given $\rvy^k$, where $\rank(\bX)\le r$.  We will also provide the functional form of the reliability function (error exponent) for this recovery problem. To do so we first consider the following optimization problem:
\begin{align}
&\minimize \,\quad \rank(\tilde{ \bX})  \nn\\*
&\st \quad\langle \rvbH_a,\tilde{ \bX}\rangle=\rvy_a,\quad a\in[k] \label{eqn:minrank}
\end{align}
The optimization variable is $\tilde{\bX} \in\bbF_q^{n\times n}$. Thus among all the matrices  that satisfy the linear  constraints in~\eqref{eqn:linear_meas}, we select one whose rank is the smallest. We call the optimization problem in~\eqref{eqn:minrank} the {\em min-rank decoder},  denoting the set of minimizers  as $\calS \subset\bbF_q^{n\times n}$. If $\calS$ is a singleton set, we also denote the unique optimizer to \eqref{eqn:minrank}, a random quantity, as $\rvbX^*$. We analyze the error probability that either $\calS$ is not a singleton set or $\rvbX^*$ does not equal   the true matrix $\bX$,  i.e., the error event
\begin{equation}
\calE_n := \{|\calS|>1\}\cup (\{|\calS|=1\} \cap \{\, \rvbX^*\ne \bX\} ).  \label{eqn:err_event}
\end{equation}
The optimization   in~\eqref{eqn:minrank} is, in general, intractable (in fact NP-hard) unless there is additional structure on the sensing matrices $\rvbH_a$ (See discussions in Section~\ref{sec:concl}). Our focus, in this paper, is on the information-theoretic limits for solving~\eqref{eqn:minrank} and its variants.  We remark that the minimization problem is  reminiscent of Csisz\'{a}r's so-called $\alpha$-decoder for linear codes~\cite{Csi82}. In~\cite{Csi82},  Csisz\'{a}r analyzed the error exponent of  the decoder that minimizes a function $\alpha(\fndot)$ [e.g., the  entropy $H(\fndot)$] of the type (or empirical distribution) of a sequence subject to the sequence satisfying a set of linear  constraints.

For this  section and Section~\ref{sec:noisy}, we assume that  each element in each sensing matrix is drawn independently and uniformly at random from $\bbF_q$, i.e., from the pmf
\begin{equation}
P_{\rvh}(h; q)=1/q,\qquad \forall \, h\in \bbF_q.\label{eqn:uar}
\end{equation} 
We call this   the {\em uniform} or {\em equiprobable} measurement model. For simplicity, throughout this section, we use the notation $\bbP$ to denote the probability measure associated to the  equiprobable measurement model.

\subsection{A Sufficient Condition for Recovery in the Noiseless Case}
In this subsection, we assume the noiseless linear model in~\eqref{eqn:linear_meas}.      We can now exploit ideas from~\cite{Dra09} to demonstrate the following achievability (weak recovery) result. Recall that $\bX$ is non-random and fixed, and we are asking how many measurements $\rvy_1,\ldots, \rvy_k$ are sufficient for recovering $\bX$. 
\begin{proposition}[Achievability] \label{prop:uniform}
Fix $\veps>0$. Under the uniform measurement model as in~\eqref{eqn:uar}, if 
\begin{equation}
k>(2+\veps)\gamma \left(1-\gamma/2\right)n^2\label{eqn:ach}
\end{equation}
then $\bbP(\calE_n)\to 0$ as $n\to\infty$. 
\end{proposition}
Note that the number of measurements stipulated by Proposition~\ref{prop:uniform} matches the information-theoretic lower bound in~\eqref{eqn:conversek}. In this sense, the min-rank decoder prescribed by the optimization problem in~\eqref{eqn:minrank} is asymptotically optimal, i.e., the bounds are {\em sharp}.  Note also  that in the converse   (Proposition~\ref{prop:converse}), the range of the decoder $\hat{\rvbX}(\fndot)$ is constrained to be the set of matrices whose rank does not exceed $r$. Hence, the decoder in the converse   has additional side information -- namely the upper bound on the rank. For the min-rank decoder in~\eqref{eqn:minrank}, no such knowledge of the rank is required and yet it meets the lower bound.  We remark that the packing-like achievability proof  is much simpler than   the typicality-based argument presented by Vishwanath in~\cite{Vis10} (albeit in a   different setting).

\begin{proof}
To each matrix  $\bZ\in \bbF_q^{n\times n}$ that is not equal to $\bX$ and whose rank is no greater than $\rank(\bX)$, define the event 
\begin{equation}
\calA_{\bZ}:= \{ \langle \bZ, \rvbH_a\rangle=\langle \bX, \rvbH_a\rangle, \forall\, a\in [k]\}. \label{eqn:AZ}
\end{equation}
Then we note that  
\begin{equation} \label{eqn:error_union}
\bbP(\calE_n)=\bbP\left(\bigcup_{\bZ:\bZ\ne\bX,\rank(\bZ)\le \rank(\bX)}\calA_{\bZ} \right) 
\end{equation}
since an error occurs if and only if there exists a matrix $\bZ\ne\bX$ such that (i) $\bZ$ satisfies the linear constraints, (ii) its rank is less than or equal to the rank of $\bX$. Furthermore, we claim that $\bbP(\calA_{\bZ})=q^{-k}$ for every $\bZ\ne \bX$. This follows because 
\begin{align}
\bbP(\calA_{\bZ}) &= \bbP(\langle \bZ-\bX, \rvbH_a\rangle=0,a\in [k])\nn\\
&\eqa \bbP(\langle \bZ-\bX ,\rvbH_1\rangle=0)^k \eqb q^{-k} , \label{eqn:probA}
\end{align}
where $(a)$ follows from the fact that the $\rvbH_a$ are i.i.d.\ matrices and $(b)$ from the fact $\bZ-\bX\ne \bzero$ and every non-zero element in a finite field has a (unique)  multiplicative inverse so $\bbP(\langle\bZ-\bX,\rvbH_1\rangle=0)=q^{-1}$~\cite{Csi82,Dra09}. More precisely, this is because     $\langle \bZ-\bX, \rvbH_1\rangle$ has distribution $P_{\rvh}$ by independence and uniformity of the elements in $\rvbH_1$. Since $r/n\to\gamma$, for any fixed $\eta'>0$, $|r/n-\gamma|\le \eta'$ for all $n$ sufficiently large. By the uniform continuity of the function $t\mapsto 2t-t^2$ on $t\in [0,1]$, for any $\eta>0$,  $|(2 nr-r^2)/n^2-2\gamma(1-\gamma/2)|\le \eta$ for all   $n\ge N_\eta$ (an integer just depending on $\eta$).  Now by combining~\eqref{eqn:probA} with the union of events bound,  
\begin{align}
\bbP(\calE_n)&\le\sum_{\bZ:\bZ\ne\bX,\rank(\bZ)\le\rank(\bX)}q^{-k}   \lec  \Psi_q(n,r)\, q^{-k} \nn\\
& \led  4 q^{2 nr-r^2 -k}   \lee 4q^{-n^2[ -2\gamma(1-\gamma/2)- \eta+k/n^2]},  \label{eqn:union1}
\end{align}
where $(c)$ follows because $\rank(\bX)\le r$,  $(d)$ follows from the upper bound in~\eqref{eqn:less_r} and $(e)$ follows for all $n$ sufficiently large as argued above. Thus, we see that if $k$ satisfies~\eqref{eqn:ach}, the exponent    in~\eqref{eqn:union1} is positive if we choose $\eta'$ sufficiently small so that  $\eta<\veps\gamma(1-\gamma/2)$. Hence $\bbP(\calE_n)\to 0$ as desired.  \end{proof}

{\em Remark:} Here and in the following, we can, without loss of generality, assume that $r=\floor{\gamma n}$ (in place of $r/n\to \gamma$). In this way, we can remove the effect of the small positive constant  $\eta$ as in the above argument.    This simplification does not affect the precision of any of the arguments in the sequel. 

\subsection{The Reliability Function}   \label{sec:ER}
We have shown in the previous section that the min-rank decoder is asymptotically optimal in the sense that the number of measurements required for it to decode $\bX$ reliably with $\bbP(\calE_n)\to 0$ matches the lower bound (necessary condition) on $k$ (Proposition~\ref{prop:converse}). It is also interesting to analyze the {\em rate of decay} of $\bbP(\calE_n)$   for the min-rank decoder.  For this purpose, we define the rate $R$ of the measurement model. 
\begin{definition}
The {\em rate} of (a sequence of) linear measurement models as in~\eqref{eqn:linear_meas} is defined as 
\begin{equation}
R:=\lim_{n\to\infty} \frac{n^2-k}{n^2}=\lim_{n\to\infty} 1-\frac{k}{n^2} \label{eqn:rate_def}
\end{equation}
assuming the limit exists. Note that $R\in [0,1]$. 
\end{definition}
The use of the term {\em rate} is in direct analogy to the use of the term in coding theory. The rate of the linear code 
\begin{equation} \label{eqn:code_def} 
\scC:=\{ \bC\in\bbF_q^{n\times n}: \lrangle{\bC}{\bH_a}=0 , a\in [k]\}
\end{equation} 
is $R_n:= 1-\dim ( \mathrm{span} \{  \vect(\bH_1),\ldots, \vect(\bH_k) \})/n^2$, which is lower bounded\footnote{The lower bound is achieved when the vectors $\vect(\bH_1),\ldots,\vect(\bH_k)$ are linearly independent in $\bbF_q$. See Section~\ref{sec:dist},  and  in particular Proposition~\ref{prop:li}, for details when the sensing matrices are random.   } by $1-k/n^2$ for every $k =0,1,\ldots, n^2$. We revisit the connection of the rank minimization problem to coding theory (and in particular to rank-metric codes) in   detail in Section~\ref{sec:dist}. 
\begin{definition}
If the limit exists, the {\em reliability function} or {\em error exponent}  of the min-rank decoder~\eqref{eqn:minrank} is defined as 
\begin{equation} \label{eqn:ER}
E(R):= \lim_{n\to\infty}-\frac{1}{n^2}\log_q\bbP(\calE_n).
\end{equation}
\end{definition}
We show   in Corollary~\ref{cor:ER} that the limit in~\eqref{eqn:ER}  indeed  exists. Unlike the usual definition of the reliability function~\cite[Eq.~(5.8.8)]{Gal}, the normalization in~\eqref{eqn:ER} is $1/n^2$ since $\bX$ is an $n\times n$ matrix.\footnote{The ``block-length'' of the   code $\scC$  in~\eqref{eqn:code_def} is $n^2$.}   Also, we restrict our attention to the min-rank decoder.  The following proposition provides an upper bound on the reliability function  of the min-rank decoder when there is no noise in the measurements as in~\eqref{eqn:linear_meas}.  

\begin{proposition}[Upper bound on $E(R)$]  \label{prop:ER} Assume that $\rank(\bX)/n\to \tgamma$  as $n\to\infty$. Under the uniform measurement model in~\eqref{eqn:uar} and assuming the min-rank decoder is used,
\begin{equation}
E(R) \le  \left| (1-R)-2\tgamma\left(1-{\tgamma}/{2}\right) \right|^+.  \label{eqn:lowerER}
\end{equation}
 
\end{proposition}
The proof of this result hinges on the {\em pairwise independence} of the events $\calA_{\bZ}$ and  de Caen's inequality~\cite{deCaen}, which for the reader's convenience, we restate here:
\begin{lemma} [de Caen \cite{deCaen}]
Let $(\Omega , \scF, \bbQ)$ be a probability space. For a finite number events $\calB_1,\ldots, \calB_M\in \scF$, the probability of their union can be lower bounded as 
\begin{equation}\label{eqn:dC}
\bbQ\left(\bigcup_{m=1}^M \calB_m\right)\ge \sum_{m=1}^M \frac{\bbQ(\calB_m)^2}{\sum_{m'=1}^M \bbQ(\calB_m\cap\calB_{m'})}.
\end{equation}
\end{lemma}
We now prove Proposition \ref{prop:ER}.
\begin{proof}  
In order to apply~\eqref{eqn:dC} to analyze the error probability in~\eqref{eqn:error_union}, we need to compute the   probabilities $\bbP(\calA_{\bZ})$ and $\bbP(\calA_{\bZ}\cap \calA_{\bZ'})$. The former is $q^{-k}$ as argued in~\eqref{eqn:probA}.  The latter uses the following lemma which is proved in Appendix~\ref{app:pairwise}. 
\begin{lemma}[Pairwise Independence] \label{lem:pairwise}
For any two distinct matrices $\bZ$ and $\bZ'$, neither  of which is equal to $\bX$, the events $\calA_{\bZ}$ and $\calA_{\bZ'}$ (defined in~\eqref{eqn:AZ}) are independent. 
\end{lemma}
As a result of this lemma,  $\bbP(\calA_{\bZ}\cap \calA_{\bZ'})= \bbP(\calA_{\bZ})\bbP(\calA_{\bZ'}) =q^{-2k}$ if $\bZ\ne\bZ'$ and $\bbP(\calA_{\bZ}\cap \calA_{\bZ'})= \bbP(\calA_{\bZ})= q^{-k}$ if $\bZ=\bZ'$. Now, we apply the lower bound~\eqref{eqn:dC} to $\bbP(\calE_n)$ noting from~\eqref{eqn:error_union} that $\calE_n$ is the union of all $\calA_{\bZ}$ such that $\bZ\ne \bX$ and $\rank(\bZ)\le \tilr:=\rank(\bX)$. Then, for a fixed $\eta>0$, we have 
\begin{align}
\bbP( \calE_n)& \ge  \sum_{\substack{\bZ:\bZ\ne\bX\\ \rank(\bZ)\le \rank(\bX)}} \frac{q^{-2k} }{q^{-k}\left(1+ \sum_{\substack{\bZ':\bZ'\ne\bX,\bZ \\ \rank(\bZ')\le \rank(\bX)}} q^{-k} \right)} \nn\\
&  \gea \frac{(q^{2n\tilr-\tilr^2}-1) q^{-k}}{1+ 4 q^{2n\tilr-\tilr^2 -k}}  \geb \frac{q^{n^2 [ 2\tgamma(1-\tgamma/2)-\eta   -k/n^2 ] }-q^{-k}}{1+ 4 q^{n^2 [ 2\tgamma(1-\tgamma/2)+\eta  -k/n^2 ]  } } ,\nn 
\end{align}
where $(a)$ is from the upper and lower bounds in~\eqref{eqn:less_r} and $(b)$ holds for all $n$ sufficiently large since $\tilr/n\to\tgamma$. See argument  justifying inequality $(c)$ in \eqref{eqn:union1}. Assuming that $1-R > 2\tgamma\left(1 - {\tgamma}/{2}\right)$, the normalized logarithm of the  error probability can now be simplified as 
\begin{align}
\limsup_{n\to\infty}-\frac{1}{n^2}\log_q \bbP(\calE_n)  
 \le -2\tgamma\left(1- {\tgamma}/{2}\right)  +\eta +\lim_{n\to\infty}\frac{k}{n^2}, \label{eqn:liminf0}
\end{align}
where we used the fact that $4 q^{n^2 [ 2\tgamma(1-\tgamma/2)+\eta  -k/n^2 ]  }\to 0$ for sufficiently small $\eta>0$.  The case where $1-R\leq 2\tgamma\left(1 - {\tgamma}/{2}\right)$ results in $E(R)=0$ because $\bbP(\calE_n)$ fails to converge to zero as $n\to\infty$. The proof of the upper bound of the reliability function  is completed by appealing to the definition of $R$ in~\eqref{eqn:rate_def}  and the arbitrariness of $\eta>0$.
\end{proof}

\begin{corollary}[Reliability function] \label{cor:ER}
Under the assumptions of Proposition~\ref{prop:ER}, the error exponent of the min-rank decoder is 
\begin{equation}
E(R) = \left| (1-R)-2\tgamma\left(1-{\tgamma}/{2}\right)\right|^+    .   \label{eqn:RF}
\end{equation}
\end{corollary}
\begin{proof}
The lower bound on $E(R)$ follows from the achievability in~\eqref{eqn:union1}, which may be   strengthened as follows:
\begin{equation}
\bbP(\calE_n)\le 4q^{-n^2 \, \left|  -2\tgamma(1-\tgamma/2) -\eta+k/n^2   \right|^+ } , \label{eqn:strengthened}
\end{equation}
since $\bbP(\calE_n)$ can also be upper bounded by unity.  Now,  because $|\fndot|^+$ is continuous,  the lower limit of the normalized logarithm of the bound in~\eqref{eqn:strengthened} can be expressed as follows:
\begin{equation}
\liminf_{n\to\infty} - \frac{1}{n^2} \log_q \bbP(\calE_n)   
   \ge   \left| -2\tgamma\left(1-{\tgamma}/{2}\right) -\eta + \lim_{n\to\infty}  \frac{k} {n^2} \right|^+  . \label{eqn:liminf}
\end{equation}  
Combining the upper bound in Proposition~\ref{prop:ER} and the lower bound in~\eqref{eqn:liminf} and noting that $\eta>0$ is arbitrary yields the reliability function in~\eqref{eqn:RF}. 
\end{proof}
We observe that pairwise independence of the events $\calA_{\bZ}$ (Lemma \ref{lem:pairwise}) is essential in the proof of Proposition~\ref{prop:ER}. Pairwise independence is a consequence of the linear measurement model in~\eqref{eqn:linear_meas} and the uniformity assumption in~\eqref{eqn:uar}. Note that the events $\calA_{\bZ}$ are {\em not} jointly (nor triple-wise) independent. But the beauty of de Caen's bound allows us to exploit the pairwise independence to lower bound $\bbP(\calE_n)$ and thus to obtain a tight upper bound on $E(R)$. To draw an analogy, just as only pairwise independence is required to show that linear codes achieve capacity in symmetric DMCs, de Caen's inequality allows us to move the exploitation of pairwise independence into the error exponent domain to make statements about the error exponent behavior of ensembles of linear codes. 

A natural question arises: Is $E(R)$ given in~\eqref{eqn:RF} the largest possible exponent over all   decoders $\hat{\rvbX}(\fndot)$ for the model in which $\rvbH_a$ follows the uniform pmf? We conjecture that this is indeed the case, but a proof   remains elusive. 

\subsubsection{Comparison of error exponents to existing works~\cite{SilPersonal}}\label{sec:comparison}
As mentioned in the Introduction, the preceding results can be interpreted from a coding-theoretic perspective. This is indeed what we will do in Section~\ref{sec:dist}. In this subsection, we compare the reliability function derived in Corollary~\ref{cor:ER} with three other coding techniques present in the literature. First, we have the well-known construction of maximum rank distance (MRD) codes by Gabidulin~\cite{Gab85}. Second,  we have the error trapping technique~\cite{Sil10} alluded to in Section~\ref{sec:motivation}. Third,  we have a combination of the two preceding code constructions which is discussed in~\cite[Section VI.E]{Sil10}. To perform this comparison, we define another reliability function $E_1(R)$ that is ``normalized by $n$''. This is simply the quantity in~\eqref{eqn:ER} where the normalization is ${1}/{n}$ instead of ${1}/{n^2}$. We now denote the reliability function normalized by $n^2$ as in~\eqref{eqn:ER} by $E_2(R)$.  We also use various superscripts on $E_1$ and $E_2$ to denote different coding schemes. Hence, for our encoding and decoding strategy using random sensing   and   min-rank decoding (RSMR), $E_1^{\mathrm{RSMR}}(R)=\infty$ for all $R\le (1-\gamma)^2$ and $E_2^{\mathrm{RSMR}}(R)$ is given by~\eqref{eqn:RF}.

Since Gabidulin codes are MRD, they achieve   the Singleton bound \cite[Section III]{Gad08}) for rank-metric codes given by $n^2-k\le n(n-d_{\mathrm{R}}+1)$, where $d_{\mathrm{R}}$ is the minimum rank distance of the code  in \eqref{eqn:code_def} [See exact definitions in \eqref{eqn:min_dist} and \eqref{eqn:min_dist_alt}]. Thus, it can be verified that for $j = 1,2$,
\begin{align}
E_j^{\mathrm{Gab}}(R) = \left\{ \begin{array}{cc}
\infty &  R\le 1-2\gamma \\
0 & \mathrm{else}
\end{array}\right. . \label{eqn:gab}
\end{align}
From \cite[Section IV.B, Eq.\ (12)]{Sil10}, it can also be checked that for the error trapping coding strategy, assuming the low-rank  error matrix is uniformly distributed over those of rank $r$,
\begin{align}
E_1^{\mathrm{ET}}(R) = \left|1-\gamma-\sqrt{R} \right|^+,\qquad E_2^{\mathrm{ET}}(R) =0. \label{eqn:et}
\end{align}
Finally,  from~\cite[Section VI.E]{Sil10}, for the combination of Gabidulin coding and error trapping, under the same condition of uniformity,
\begin{align}
E_1^{\mathrm{GabET}}(R)=\left|1-\gamma- \frac{R}{1-\gamma} \right|^+,\quad E_2^{\mathrm{GabET}}(R) =0. \label{eqn:gaberror}
\end{align}
Note that for the error exponents in~\eqref{eqn:gab}, \eqref{eqn:et} and~\eqref{eqn:gaberror}, the randomness is over the low-rank error matrix $\rvbX$ and not the code construction, which is {\em deterministic}. In contrast, our coding strategy RSMR involves a  {\em random} encoding scheme.  It can be seen from~\eqref{eqn:gab} to~\eqref{eqn:gaberror} that there is a non-trivial interval of rates $\scR:=[1-2\gamma, (1-\gamma)^2]$ in which our reliability functions  $E_1^{\mathrm{RSMR}}(R)$ and $E_2^{\mathrm{RSMR}}(R)$ are the best (largest). Indeed, in the interval $\scR$, $E_1^{\mathrm{RSMR}}(R)=\infty$ and our result in~\eqref{eqn:ER} implies that $E_2^{\mathrm{RSMR}}(R)>0$ whereas all the abovementioned coding schemes give $E_2(R)=0$. Thus, using both a random code for encoding and    min-rank decoding is   advantageous from a reliability function standpoint in the regime $R\in\scR$. Furthermore, as we shall see from~\eqref{eqn:speed} in Section~\ref{sec:sparse} which deals with the sparse sensing setting (SRSMR), $E_1^{\mathrm{SRSMR}}(R)=\infty$ and $E_2^{\mathrm{SRSMR}}(R)=0$  for all $R\le (1-\gamma)^2$. Such an encoding scheme using sparse parity-check matrices may be amenable for  the design of low-complexity decoding strategies that also have good error exponent properties.  In general  though, our min-rank decoder requires exhaustive search (though Section~\ref{sec:nonexhaust} proposes techniques to reduce the search space), while all the preceding techniques     have polynomial-time decoding complexity.

\section{Uniformly Random Sensing Matrices: The Noisy Case} \label{sec:noisy}
We now generalize the noiseless model  and the accompanying results in Section~\ref{sec:uar} to the case where the measurements $\rvy^k$ are noisy as in~\eqref{eqn:noisy_meas}.  As in  Section~\ref{sec:uar}, we assume that the elements of $\rvbH_a$ are i.i.d.\ and uniform in $\bbF_q$.  The noise $\rvbw$ is first assumed  in Section~\ref{sec:det} to be   deterministic but unknown. We then extend our results  to the situation where $\rvbw$ is a random vector in Section~\ref{sec:rand}.  

\subsection{Deterministic Noise} \label{sec:det}
In the deterministic setting, we assume that $\|\bw\|_0 = \floor{\sigma n^2}$ for some {\em noise level} $\sigma\in (0,k/n^2]$. Instead of using the minimum entropy decoder as in~\cite{Dra09} (also see~\cite{Csi82}), we consider the following   generalization of the min-rank decoder:
\begin{align}
&\minimize \quad \rank(\tilde{\bX})  +\lambda\|\tilde{\bw}\|_0 \nn\\*
&\st \quad\langle \rvbH_a,\tilde{\bX}\rangle+\tilde{w}_a=\rvy_a,\quad a\in[k] \label{eqn:minrank2}
\end{align}
The optimization variables are $\tilde{\bX}\in\bbF_q^{n\times n}$ and $\tilde{\bw}\in\bbF_q^{k}$. The parameter $\lambda=\lambda_n>0$ governs the tradeoff between the rank of the matrix $\bX$ and the sparsity of the vector $\bw$.  Let   $\Hq(p):=-p \log_q (p)-(1-p)\log_q (p)$ be the  (base-$q$) binary entropy. 
\begin{proposition}[Achievability under deterministic noisy measurement model] \label{prop:noisy}
Fix $\veps>0$ and choose $\lambda=1/n$. Assume the uniform measurement model and that $\|\bw\|_0= \floor{\sigma n^2}$.  If 
\begin{equation}
k>\frac{(3+\veps)(\gamma+\sigma)[1-(\gamma+\sigma)/3]}{1-H_2[1/ (3-(\gamma+\sigma))  ]     \log_q 2}  \,\, n^2 , \label{eqn:knoisy}
\end{equation}
then $\bbP(\calE_n)\to 0$ as $n\to\infty$. 
\end{proposition}
The proof of this proposition is provided in Appendix~\ref{app:noisy1}. Since the prefactor in~\eqref{eqn:knoisy} is a monotonically increasing function in the noise level $\sigma$, the number of measurements increases as  $\sigma$ increases, agreeing with intuition. Note that the  regularization parameter $\lambda$ is chosen to be $1/n$ and  is thus independent of $\sigma$. Hence, the decoder does not need to know the true value of the noise level $\sigma$. The factor of~3 (instead of~2) in \eqref{eqn:knoisy} arises in part due  to the uncertainty in the locations of the non-zero elements of the noise vector $\bw$. We remark that Proposition~\ref{prop:noisy} does not reduce to the noiseless case ($\sigma=0$)  in Proposition~\ref{prop:uniform} because we assumed a different measurement  model in~\eqref{eqn:noisy_meas}, and employed a different bounding technique. 

The measurement complexity in~\eqref{eqn:knoisy} is suboptimal, i.e., it does not match the  converse  in~\eqref{eqn:conversek}. This is because the decoder in~\eqref{eqn:minrank2}   estimates   {\em both} the matrix $\bX$ {\em and} the noise $\bw$  whereas in the derivation of the converse, we are only concerned with reconstructing the unknown matrix $\bX$. By decoding $(\bX,\bw)$ jointly, the analysis proceeds along the lines of the proof of Proposition~\ref{prop:uniform}. It is unclear whether a better parameter-free decoding strategy exists in the presence of noise and whether such a  strategy is also amenable to analysis. The noisy setting was also analyzed in~\cite{Emad11} but, as in our work, the number of measurements for achievability does not match the converse. 
\subsection{Random Noise}\label{sec:rand}
We now consider the case where the noise   in~\eqref{eqn:noisy_meas} is random, i.e., $\rvbw= (\rvw_1,\ldots, \rvw_k) \in \bbF_q^k$ is a {\em random} vector.      We assume the noise vector $\rvbw$ is i.i.d.\  and each component is distributed according to {\em any}  pmf   for which
\begin{equation}
P_{\rvw}(w;p)=1-p \qquad\mbox{if}\qquad w=0 . \label{eqn:wdist}
\end{equation}
This pmf represents a noisy channel where every symbol is changed to some other (different) symbol   independently with crossover probability $p \in (0,1/2)$.   We can ask how many measurements are necessary and sufficient for recovering a fixed $\bX$ in the presence of the additive stochastic noise $\rvbw$. Also,  we are interested to know how this measurement complexity depends on $p$.   We leverage on Propositions~\ref{prop:converse} and~\ref{prop:noisy} to derive a converse result and an achievability result respectively. We start with the converse, which is partially inspired by Theorem~3 in \cite{Emad11}. 

\begin{corollary}[Converse under random noise model] \label{cor:conv_noise} Assume the setup in Proposition~\ref{prop:converse} and consider the noisy measurement model given by~\eqref{eqn:noisy_meas} and~\eqref{eqn:wdist}. Additionally,  assume that $\rvbX$, $\rvbH^k$ and $\rvbw$ are jointly independent. If, 
\begin{equation}
k<\frac{(2-\veps)\gamma \left(1-\gamma/2\right)}{1-H_q(p)} n^2\label{eqn:conversek_noise}
\end{equation}
then for any estimator,  $\bbP(\tilde{\calE}_n)\ge {\veps}/{4} > 0$ for all $n$ sufficiently large, where $\tilde{\calE}_n$ is defined in \eqref{eqn:tildeE}. 
\end{corollary} 
Note that the probability of error $\bbP(\tilde{\calE}_n)$ above is computed over both the randomness in the sensing matrices $\rvbH_a$ and in the noise $\rvbw$.   The proof is given in Appendix~\ref{app:noisy_conv}. From~\eqref{eqn:conversek_noise},  the number of measurements necessarily has to increase by a factor of $1/(1-H_q(p))$ for reliable recovery. As expected, for a fixed $q$, the larger the crossover probability $p\in (0,1/2)$, the more measurements are required. The converse is illustrated  for different parameter settings in  Figs.~\ref{fig:noisy} and~\ref{fig:noisy2}.

To present our achievability result compactly, we assume that $k= \ceil{\alpha n^2}$ for some  scaling parameter $\alpha \in (0,1)$, i.e., the number of observations is proportional to $n^2$ and the constant of proportionality is $\alpha$. We would like to find the  range of values of the {\em scaling parameter} $\alpha$ such that reliable recovery is possible. Recall that the upper bound on the rank is $r$ and the  noise vector has expected weight $p k\approx p \alpha n^2$. 

\begin{corollary}[Achievability under random noisy measurement model] \label{prop:noisy2}
Fix $\veps>0$  and choose    $\lambda = 1/n$.  Assume the uniform measurement model and that $k=\ceil{\alpha n^2}$.  Define   the function
\begin{equation} 
g(\alpha;p, \gamma)\!:= \! \alpha \big[ 1-(\log_q 2)H_2 \left(p+ {\gamma }/{\alpha}   \right)    - 2p (1-\gamma)\big]\!+\!\alpha^2 p^2 .  \label{eqn:defg}
\end{equation}
If  the tuple $(\alpha,p, \gamma)$ satisfies the following inequality:
\begin{align}
g(\alpha; p, \gamma)\ge (2+\veps)\gamma(1-\gamma/2), \label{eqn:alpha_ineq}
\end{align}
then $\bbP(\calE_n)\to 0$ as $n\to\infty$. 
\end{corollary}
The proof of this corollary uses typicality arguments and is presented in Appendix~\ref{app:noisy2}. As in the deterministic noise setting, the sufficient condition in~\eqref{eqn:alpha_ineq} does not reduce to the noiseless case ($p=0$) in Proposition~\ref{prop:uniform}. It also does not match the converse in~\eqref{eqn:conversek_noise}. This is due to the different bounding technique employed to prove Corollary~\ref{prop:noisy2} [both $\bX$ and $\rvbw$ are decoded in~\eqref{eqn:minrank2}].     In addition, the inequality in~\eqref{eqn:alpha_ineq} does not admit an analytical solution for $\alpha$. Hence, we search for the critical $\alpha$ [the minimum one satisfying~\eqref{eqn:alpha_ineq}] numerically for some parameter settings. See Figs.~\ref{fig:noisy} and~\ref{fig:noisy2} for illustrations of how the critical $\alpha$  varies with $(p,\gamma)$ when  the field size is small ($q=2$) and when it is large  ($q=256$).

\begin{figure}
\centering
\includegraphics[width = \linewidth]{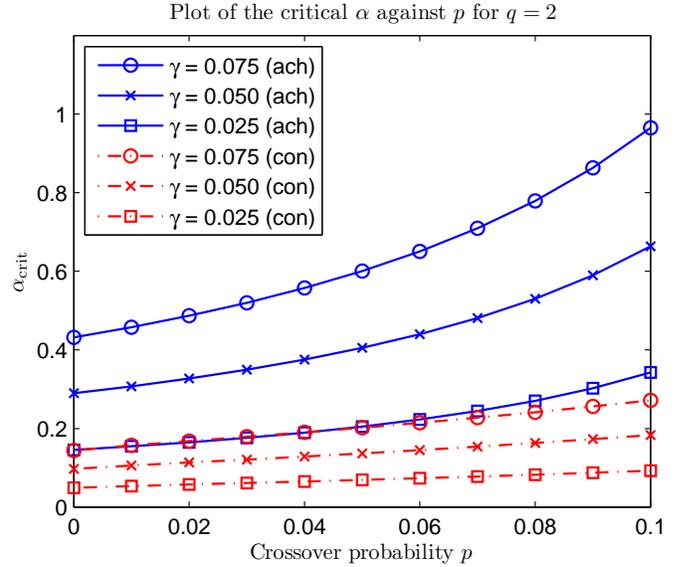}
\caption{Plot of $\alpha_{\mathrm{crit}}$   against $p$   for  $q=2$. Both  $\alpha_{\mathrm{crit}}$ for the converse (con) in~\eqref{eqn:conversek_noise} the achievability (ach) in~\eqref{eqn:alpha_ineq} are shown.  All $\alpha$'s below the converse curves are not achievable.  }
\label{fig:noisy}
\end{figure}

\begin{figure}
\centering
\includegraphics[width = \linewidth]{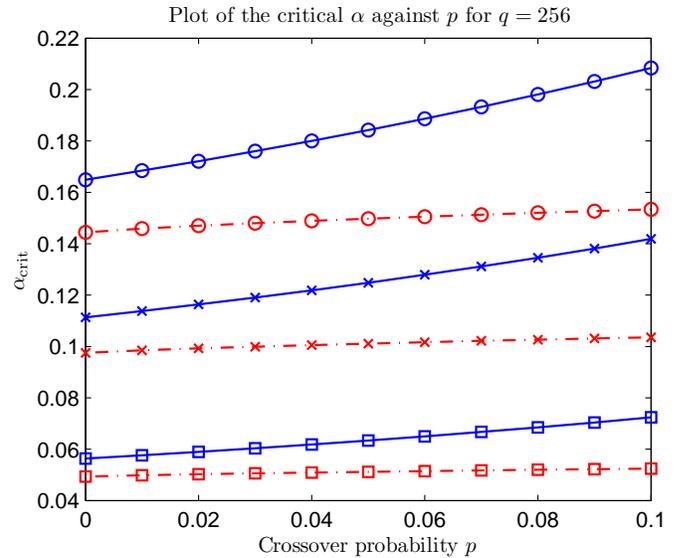}
\caption{Plot of  $\alpha_{\mathrm{crit}}$  against $p$  for $q=256$. See Fig.~\ref{fig:noisy} for the legend.  }
\label{fig:noisy2}
\end{figure}

From   Fig.~\ref{fig:noisy}, we observe that the noise results in a  significant increase in the critical value of the scaling parameter $\alpha$ when $q=2$.  We see that for a rank-dimension ratio of $\gamma = 0.05$ and with a crossover probability of $p=0.02$, the critical scaling parameter  is $\alpha_{\mathrm{crit}}\approx 0.32$. Contrast this to the noiseless case (Proposition~\ref{prop:uniform}) and the converse result for the noisy case (Corollary~\ref{cor:conv_noise}) which stipulate  that the critical scaling parameters are $2\gamma(1-\gamma/2)\approx 0.098$ and $2\gamma(1-\gamma/2)/(1-H_2(p))\approx 0.114$ respectively. Hence,   we incur roughly a threefold increase in the number of measurements to tolerate  a noise level of $p=2\%$.  This phenomenon is  due to   our incognizance of the locations of the non-zero elements of $\rvbw$ (and hence knowledge of  which measurements $\rvy_a$ are reliable).  In contrast to the reals, in  the finite field setting, there is no notion of  the ``size'' of the  noise (per measurement). Hence, estimation performance in the presence of noise does not degrade as gracefully as in the reals (cf.~\cite[Theorem~1.2]{Meka}).   However, when the field size is large (more likened to the reals), the degradation is not as severe. This is depicted in Fig.~\ref{fig:noisy2}. Under the same settings as above, $\alpha_{\mathrm{crit}}\approx  0.114$, which is not too far from the converse ($2\gamma(1-\gamma/2)/(1-H_{256}(p))\approx 0.099$).

As a final remark, we compare the decoders for the noisy model in~\eqref{eqn:minrank2} and that in~\cite{Emad11}. In~\cite{Emad11}, the authors considered the (analog of) following decoder (for tensors):
\begin{align}
&\minimize \quad\,\, \rank(\tilde{\bX})   \nn\\*
&\st  \quad\|\, \rvby_{\tilde{\bX}} -  \rvby\, \|_0\le\tau , \label{eqn:minrank3}
\end{align}
where $\rvby_{\tilde{\bX}}:=[ \lrangle{\rvbH_1}{\tilde{\bX}} \,\ldots\, \lrangle{\rvbH_k}{\tilde{\bX}} ]^T$ and $\rvby=\rvy^k$ is the   noisy observation vector in \eqref{eqn:noisy_meas}.  However,  the threshold $\tau$   that constrains the Hamming distance between $\rvby_{\tilde{\bX}}$ and $\rvby$  is not straightforward to choose.\footnote{In fact,  the achievability result of Theorem~4 in~\cite{Emad11} says that $\tau=\eta k$ where $\eta\in (p, (q-1)/q)$ but   for our optimization program in~\eqref{eqn:minrank2}, the decoder does not need to know the crossover probability $p$. } Our decoder, in contrast, is  parameter-free because the regularization constant $\lambda$ in~\eqref{eqn:minrank2}  can be chosen to be $1/n$, {\em independent} of all other parameters.  In addition, Fig.~\ref{fig:compq} shows that at high $q$, our decoder and analysis result in a better (smaller) $\alpha_{\mathrm{crit}}$ than that in~\cite{Emad11}. Our decoding scheme gives a bound that is   closer to the converse at high $q$ while the decoding scheme in~\cite{Emad11} is farther. The slight disadvantage of our decoder is that the number of measurements   in~\eqref{eqn:alpha_ineq} cannot be expressed in closed-form.

\begin{figure}
\centering
\includegraphics[width = \linewidth]{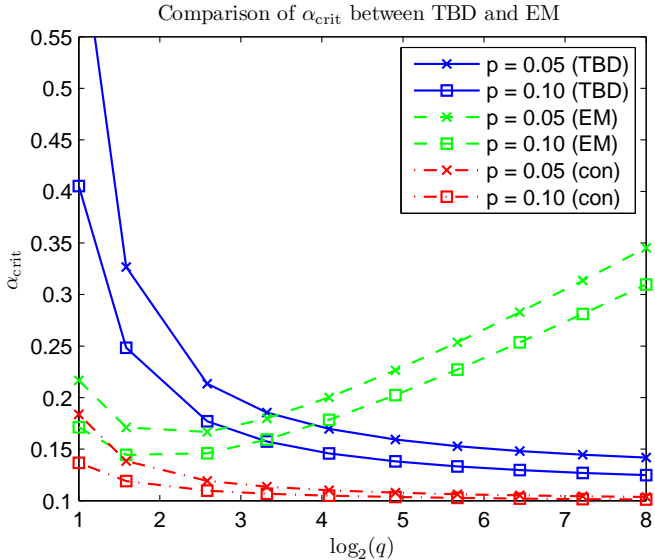}
\caption{Plot of  $\alpha_{\mathrm{crit}}$  against   $\log_2(q)$ for our work (TBD  Corollary~\ref{prop:noisy2}), the converse in Corollary~\ref{cor:conv_noise} and  Emad and Milenkovic (EM)~\cite{Emad11}.  }
\label{fig:compq}
\end{figure}

\section{Sparse Random Sensing Matrices} \label{sec:sparse}
In the previous two sections, we focused exclusively on the case where the elements of the sensing matrices $\rvbH_a,a\in [k],$ are drawn uniformly from $\bbF_q$. However, there is substantial motivation to consider other ensembles of sensing matrices. For example, in low-density parity-check (LDPC) codes, the parity-check matrix (analogous to the set of $\rvbH_a$ matrices) is sparse. The sparsity aids in decoding via the sum-product algorithm~\cite{Ksc01} as the resulting Tanner (factor) graph is sparse~\cite{Bar10}. In~\cite{Kak11}, the authors considered the case where the generator matrices are sparse and random but their setting is restricted to the BSC and BEC channel models.

In this section, we revisit the noiseless model in~\eqref{eqn:linear_meas} and analyze the scenario where the sensing matrices are sparse on average.  More precisely,  each element of   $\rvbH_a,a\in [k],$ is assumed to be an i.i.d.\ random variable with associated pmf 
\begin{equation}
P_{\rvh}(h;\delta,q):= \left\{ \begin{array}{cc}
1-\delta & h=0\\
\delta/(q-1) & h\in \bbF_q\setminus\{0\}\\
\end{array}
\right.  . \label{eqn:prob_delta}
\end{equation}
Note that  if $\delta$ is small, then the probability that an entry in $\rvbH_a$ is zero is close to unity. The problem of deriving a sufficient condition for reliable recovery is more challenging as compared to the equiprobable case since~\eqref{eqn:probA} no longer holds (compare to  Lemma~\ref{lem:circconv}). Roughly speaking, the matrix $\bX$ is not sensed as much as  in the equiprobable case and   the measurements $\rvy^k$ are not as informative because $\rvbH_a,a\in [k]$, are sparse. In the rest of this section, we allow the sparsity factor $\delta$ to depend on $n$ but we do not make the dependence of $\delta$ on $n$ explicit for ease of exposition. The question we would like to answer is: {\em How fast can $\delta$ decay with $n$ such that the min-rank decoder is still reliable for  weak recovery?  }
\begin{theorem}[Achievability under sparse measurement model]  \label{prop:sparse}
Fix $\veps>0$ and let $\delta$  be any    sequence in   $\Omega(\frac{\log n}{n}) \cap o(1)$.  Under the sparse measurement model as in~\eqref{eqn:prob_delta},  if the number of measurements $k$ satisfies~\eqref{eqn:ach}  for all $n>N_{\veps,\delta} $, then $\bbP(\calE_n)\to 0$ as $n\to\infty$. 
\end{theorem}
The proof of Theorem~\ref{prop:sparse}, our main result,  is detailed in Appendix~\ref{app:sparse}.  It utilizes a ``splitting'' technique to partition the set of  {\em misleading} matrices $\{\bZ\ne\bX:\rank(\bZ)\le\rank(\bX)\}$ into those with low Hamming distance from $\bX$ and those with high Hamming distance from $\bX$.

Observe that the sparsity-factor $\delta$  is allowed to tend to zero albeit at a controlled rate of $\Omega(\frac{\log n}{n})$.  Thus, each $\rvbH_a$ is allowed to have, on average, $\Omega(n \log n)$ non-zero entries (out of  $n^2$ entries). The scaling rate is reminiscent of the number of trials required for success in the so-called {\em coupon collector's problem}. Indeed, it seems plausible that we need at least one entry in each row and one entry in each column of $\bX$ to be sensed (by  a sensing matrix $\rvbH_a$) for the min-rank decoder to succeed. It can easily be seen that if $\delta=o(\frac{\log n}{n})$,  there will be at least one row and one column in $\rvbH_a$ of zero Hamming weight w.h.p. Really surprisingly, the number of measurements  required in the $\delta=\Omega(\frac{\log n}{n})$-sparse sensing case is exactly the same as in the case where the elements of $\rvbH_a$ are drawn {\em uniformly} at random  from $\bbF_q$ in Proposition~\ref{prop:uniform}.  In fact it also matches  the information-theoretic lower bound in Proposition~\ref{prop:converse} and hence is asymptotically optimal. We will analyze this  weak recovery sparse setting (and understand why it works) in greater detail by studying minimum distance properties of sparse parity-check rank-metric codes in Section~\ref{sec:sparse_code}.    The sparse scenario may be extended to the noisy case  by combining  the proof techniques in   Proposition~\ref{prop:noisy} and Theorem~\ref{prop:sparse}.

There are two natural questions at this point: Firstly, can the reliability function  be computed for the min-rank decoder assuming the sparse measurement model? The events  $\calA_{\bZ}$, defined in~\eqref{eqn:AZ}, are  no longer pairwise independent.  Thus, it is not straightforward to compute  $\bbP(\calA_{\bZ}\cap\calA_{\bZ'})$ as in the proof of Proposition~\ref{prop:ER}. Further,  de Caen's lower bound may not be tight as in the case where the entries of the sensing matrices are drawn uniformly at random from $\bbF_q$.   Our bounding technique for Theorem~\ref{prop:sparse} only ensures that 
\begin{equation}
\limsup_{n\to\infty}\frac{1}{n\log n}\log_q \bbP(\calE_n)\le -C \label{eqn:speed}
\end{equation}
for some  non-trivial $C\in (0,\infty)$. Thus, instead of having a speed\footnote{The term {\em speed} is in direct analogy to the theory of large-deviations~\cite{Dembo}  where     $\bbP_n$ is said to satisfy a large-deviations upper bound with {\em speed} $a_n$ and {\em rate function} $J(\fndot)$ if $\limsup_{n\to\infty}a_n^{-1}\log\bbP_n(\calE)\le-\inf_{x\in \mathrm{cl}(\calE)}J(x)$.}	 of $n^2$ in the large-deviations upper bound, we   have a speed of $n\log n$. This is because $\delta$ is allowed to decay to zero. Whether the speed $n\log n$ is   optimal is open.   Secondly, is $\delta=\Omega(\frac{\log n}{n})$ the best (smallest) possible sparsity factor? Is there a fundamental tradeoff between the sparsity factor $\delta$  and (a bound on) the number of measurements $k$? We leave these for further research.

\input{coding_v2} 
\input{nonexhaust_v2}
\input{conclusion_v2}

\subsection*{Acknowledgements }
We would like to thank Associate Editor Erdal Ar{\i}kan  and the reviewers for their suggestions to improve the paper and to acknowledge discussions  with  Ron  Roth, Natalia Silberstein and especially Danilo Silva, who  made the insightful points in Section \ref{sec:comparison} \cite{SilPersonal}. We  would also like to thank Ying Liu and Huili Guo for detailed comments and  help in generating Fig.~\ref{fig:disjoint} respectively.
\appendices

\input{pairwise_v2}
\input{noisy_v2}
\input{noise_conv}
\input{noisy2_v2}
\input{sparse_prf_v2}
\input{stats_v2}
 \input{li_prf_v4}%

\bibliographystyle{ieeetr}
\bibliography{isitbib}

\begin{IEEEbiographynophoto}{\bf Vincent Y. F. Tan}   received the B.A.\ and M.Eng.\ degrees in    Electrical and Information Engineering    from Sidney Sussex College, Cambridge University in 2005.  He received the Ph.D.\ degree in Electrical Engineering and Computer Science (EECS) from the Massachusetts Institute of Technology (MIT)   in 2011. He is currently a postdoctoral researcher in the Electrical and Computer Engineering Department at the University of Wisconsin (UW), Madison as well as a research affiliate at the Laboratory for Information and Decision Systems (LIDS) at MIT. He  has held summer research internships at Microsoft Research in 2008 and 2009.  His  research is supported by A*STAR, Singapore. His research interests   include network information  theory, detection and estimation, and learning and inference of graphical models.

Dr.\ Tan  is a recipient of the 2005 Charles Lamb Prize, a Cambridge University Engineering Department prize awarded annually to the top candidate  in Electrical and Information Engineering. He   also received the 2011 MIT EECS Jin-Au Kong outstanding doctoral thesis prize.  He has served as a reviewer for the IEEE Transactions on Signal Processing, the IEEE Transactions on Information Theory, and the Journal of Machine Learning Research. 
\end{IEEEbiographynophoto}
  \begin{IEEEbiographynophoto}{\bf   Laura Balzano} is a Ph.D.\ candidate in Electrical and Computer Engineering, working with Professor Robert Nowak at the University of Wisconsin (UW), Madison, degree expected May 2012. Laura received her B.S.\ and M.S.\ in Electrical Engineering from Rice University 2002 and the University of California in Los Angeles 2007 respectively. She received the Outstanding M.S.\ Degree of the year award from UCLA. She has worked as a software engineer at Applied Signal Technology, Inc. Her Ph.D.\ is being supported by a 3M fellowship. Her main research focus is on low-rank modeling for inference and learning with highly incomplete or corrupted data, and its applications to communications, biological, and sensor networks, and collaborative filtering. 
  \end{IEEEbiographynophoto}
\begin{IEEEbiographynophoto}{\bf Stark C. Draper} (S'99-M'03)  is an Assistant Professor of Electrical and
Computer Engineering at the University of Wisconsin (UW), Madison. He received
the M.S.\ and Ph.D.\ degrees in Electrical Engineering and Computer Science
from the Massachusetts Institute of Technology (MIT), and the B.S. and B.A.
degrees in Electrical Engineering and History, respectively, from Stanford
University.

Before moving to Wisconsin, Dr.\ Draper worked at the Mitsubishi Electric
Research Laboratories (MERL) in Cambridge, MA. He held postdoctoral positions in the Wireless Foundations, University of California, Berkeley, and in
the Information Processing Laboratory, University of Toronto, Canada. He has
worked at Arraycomm, San Jose, CA, the C.~S Draper Laboratory, Cambridge,
MA, and Ktaadn, Newton, MA. His research interests include communication
and information theory, error-correction coding, statistical signal processing
and optimization, security, and application of these disciplines to computer
architecture and semiconductor device design.

Dr.\ Draper has received an NSF CAREER Award, the UW ECE Gerald
Holdridge Teaching Award, the MIT Carlton E. Tucker Teaching Award,
an Intel Graduate Fellowship, Stanford's Frederick E. Terman Engineering
Scholastic Award, and a U.S. State Department Fulbright Fellowship.
\end{IEEEbiographynophoto}

\end{document}

%% file: preamble.tex
\usepackage[mathscr]{eucal}
\usepackage[cmex10]{amsmath}
\usepackage{epsfig,epsf,psfrag}
\usepackage{amssymb,amsmath,amsthm,amsfonts,latexsym}
\usepackage{amsmath,graphicx,bm,xcolor,url}
\usepackage[caption=false]{subfig} 
\usepackage{fixltx2e}
\usepackage{array}
\usepackage{verbatim}
\usepackage{bm}
\usepackage{algorithmic, cite}
\usepackage{algorithm}
\usepackage{verbatim}
\usepackage{textcomp}
\usepackage{mathrsfs}
\usepackage{hyperref}

\catcode`~=11 \def\UrlSpecials{\do\~{\kern -.15em\lower .7ex\hbox{~}\kern .04em}} \catcode`~=13 

\allowdisplaybreaks[2]
 
\newcommand{\nn}{\nonumber}

\newcommand{\calA}{\mathcal{A}}
\newcommand{\calB}{\mathcal{B}}

\newcommand{\calE}{\mathcal{E}}

\newcommand{\calI}{\mathcal{I}}

\newcommand{\calK}{\mathcal{K}}
\newcommand{\calL}{\mathcal{L}}
\newcommand{\calM}{\mathcal{M}}

\newcommand{\calR}{\mathcal{R}}
\newcommand{\calS}{\mathcal{S}}
\newcommand{\calT}{\mathcal{T}}
\newcommand{\calU}{\mathcal{U}}

\newcommand{\ba}{\mathbf{a}}
\newcommand{\bA}{\mathbf{A}}
\newcommand{\bb}{\mathbf{b}}
\newcommand{\bB}{\mathbf{B}}
\newcommand{\bc}{\mathbf{c}}
\newcommand{\bC}{\mathbf{C}}

\newcommand{\be}{\mathbf{e}}

\newcommand{\bF}{\mathbf{F}}

\newcommand{\bH}{\mathbf{H}}

\newcommand{\bI}{\mathbf{I}}

\newcommand{\bM}{\mathbf{M}}

\newcommand{\bp}{\mathbf{p}}

\newcommand{\bQ}{\mathbf{Q}}
\newcommand{\br}{\mathbf{r}}
\newcommand{\bR}{\mathbf{R}}
\newcommand{\bs}{\mathbf{s}}

\newcommand{\bT}{\mathbf{T}}
\newcommand{\bu}{\mathbf{u}}
\newcommand{\bU}{\mathbf{U}}
\newcommand{\bv}{\mathbf{v}}
\newcommand{\bV}{\mathbf{V}}
\newcommand{\bw}{\mathbf{w}}

\newcommand{\bx}{\mathbf{x}}
\newcommand{\bX}{\mathbf{X}}

\newcommand{\bZ}{\mathbf{Z}}

\newcommand{\bbC}{\mathbb{C}}

\newcommand{\bbE}{\mathbb{E}}
\newcommand{\bbF}{\mathbb{F}}

\newcommand{\bbI}{\mathbb{I}}

\newcommand{\bbN}{\mathbb{N}}

\newcommand{\bbP}{\mathbb{P}}
\newcommand{\bbQ}{\mathbb{Q}}
\newcommand{\bbR}{\mathbb{R}}

\newcommand{\bbZ}{\mathbb{Z}}

\newcommand{\scC}{\mathscr{C}}

\newcommand{\scF}{\mathscr{F}}

\newcommand{\scR}{\mathscr{R}}

\DeclareMathAlphabet{\mathbsf}{OT1}{cmss}{bx}{n}
\DeclareMathAlphabet{\mathssf}{OT1}{cmss}{m}{sl}

\newcommand{\rvbA}{\mathbsf{A}}

\newcommand{\rvbC}{\mathbsf{C}}
\newcommand{\rvd}{\mathsf{d}}

\newcommand{\rvh}{\mathsf{h}}

\newcommand{\rvbh}{\mathbsf{h}}
\newcommand{\rvbH}{\mathbsf{H}}

\newcommand{\rvm}{\mathsf{m}}

\newcommand{\rvbM}{\mathbsf{M}}

\newcommand{\rvN}{\mathsf{N}}

\newcommand{\rvR}{\mathsf{R}}

\newcommand{\rvw}{\mathsf{w}}

\newcommand{\rvbw}{\mathbsf{w}}

\newcommand{\rvbX}{\mathbsf{X}}
\newcommand{\rvy}{\mathsf{y}}

\newcommand{\rvby}{\mathbsf{y}}

\DeclareSymbolFont{bsfletters}{OT1}{cmss}{bx}{n}  
\DeclareSymbolFont{ssfletters}{OT1}{cmss}{m}{n}
\DeclareMathSymbol{\bsfGamma}{0}{bsfletters}{'000}
\DeclareMathSymbol{\ssfGamma}{0}{ssfletters}{'000}
\DeclareMathSymbol{\bsfDelta}{0}{bsfletters}{'001}
\DeclareMathSymbol{\ssfDelta}{0}{ssfletters}{'001}
\DeclareMathSymbol{\bsfTheta}{0}{bsfletters}{'002}
\DeclareMathSymbol{\ssfTheta}{0}{ssfletters}{'002}
\DeclareMathSymbol{\bsfLambda}{0}{bsfletters}{'003}
\DeclareMathSymbol{\ssfLambda}{0}{ssfletters}{'003}
\DeclareMathSymbol{\bsfXi}{0}{bsfletters}{'004}
\DeclareMathSymbol{\ssfXi}{0}{ssfletters}{'004}
\DeclareMathSymbol{\bsfPi}{0}{bsfletters}{'005}
\DeclareMathSymbol{\ssfPi}{0}{ssfletters}{'005}
\DeclareMathSymbol{\bsfSigma}{0}{bsfletters}{'006}
\DeclareMathSymbol{\ssfSigma}{0}{ssfletters}{'006}
\DeclareMathSymbol{\bsfUpsilon}{0}{bsfletters}{'007}
\DeclareMathSymbol{\ssfUpsilon}{0}{ssfletters}{'007}
\DeclareMathSymbol{\bsfPhi}{0}{bsfletters}{'010}
\DeclareMathSymbol{\ssfPhi}{0}{ssfletters}{'010}
\DeclareMathSymbol{\bsfPsi}{0}{bsfletters}{'011}
\DeclareMathSymbol{\ssfPsi}{0}{ssfletters}{'011}
\DeclareMathSymbol{\bsfOmega}{0}{bsfletters}{'012}
\DeclareMathSymbol{\ssfOmega}{0}{ssfletters}{'012}

\newcommand{\tilbQ}{\tilde{\bQ}}

\newcommand{\tilr}{\tilde{r}}

\newcommand{\tilbU}{\tilde{\bU}}

\newcommand{\tilbV}{\tilde{\bV}}

\newcommand{\tilbX}{\tilde{\bX}}

\newcommand{\veps}{\varepsilon}

\newcommand{\tgamma}{\tilde{\gamma}}

\def\fndot{\, \cdot \,}

\newcommand{\ceil}[1]{\lceil{#1}\rceil}
\newcommand{\floor}[1]{\lfloor{#1}\rfloor}
\newcommand{\lrangle}[2]{\langle{#1},{#2}\rangle}

\newcommand{\dotdoteq}{\stackrel{\,..}{=}}
\newcommand{\eqa}{\stackrel{(a)}{=}}
\newcommand{\eqb}{\stackrel{(b)}{=}}
\newcommand{\eqc}{\stackrel{(c)}{=}}
\newcommand{\eqd}{\stackrel{(d)}{=}}
\newcommand{\eqe}{\stackrel{(e)}{=}}
\newcommand{\eqf}{\stackrel{(f)}{=}}

\newcommand{\lea}{\stackrel{(a)}{\le}}
\newcommand{\leb}{\stackrel{(b)}{\le}}
\newcommand{\lec}{\stackrel{(c)}{\le}}
\newcommand{\led}{\stackrel{(d)}{\le}}
\newcommand{\lee}{\stackrel{(e)}{\le}}

\newcommand{\gea}{\stackrel{(a)}{\ge}}
\newcommand{\geb}{\stackrel{(b)}{\ge}}
\newcommand{\gec}{\stackrel{(c)}{\ge}}

\DeclareMathOperator*{\argmin}{arg\,min}

\DeclareMathOperator{\minimize}{minimize}

\DeclareMathOperator{\st}{subject\,\,to}

\DeclareMathOperator{\supp}{supp}

\DeclareMathOperator{\var}{var}

\DeclareMathOperator{\rank}{rank}

\DeclareMathOperator{\vect}{vec}
\newcommand{\Hb}{H_2}
\newcommand{\Hq}{H_{q}}

\newcommand{\bzero}{\mathbf{0}}

\newtheorem{theorem}{Theorem} 
\newtheorem{lemma}[theorem]{Lemma}

\newtheorem{proposition}[theorem]{Proposition}
\newtheorem{corollary}[theorem]{Corollary}
\newtheorem{definition}{Definition}

%% file: coding_v2.tex
\section{Coding-Theoretic Interpretations and Minimum Rank Distance Properties} \label{sec:dist}
This section is devoted to understand the coding-theoretic interpretations and analogs of the rank minimization problem in~\eqref{eqn:minrank}. In particular, we would like to understand the geometry of the  random linear rank-metric codes that underpin the optimization problem in~\eqref{eqn:minrank} for both the equiprobable ensemble  in~\eqref{eqn:uar} and the sparse ensemble in~\eqref{eqn:prob_delta}. 

As mentioned in the Introduction, there is a natural correspondence between the rank minimization problem   and rank-metric decoding~\cite{Gab85, Roth,Sil08, Loi06,  Mon07 , Gad08}. In the former, we solve a problem of the form \eqref{eqn:minrank}. In the latter, the {\em code} $\scC$ typically consists of length-$n$ vectors\footnote{We abuse notation by using a common  symbol $\scC$ to denote both a code consisting of vectors with elements in $\bbF_{q^n}$ and a code consisting of matrices with elements in $\bbF_q$.  } whose elements belong to the extension field $\bbF_{q^n}$ and  these vectors  in $\bbF_{q^n}^n$ a belong to  the kernel of some linear operator $\bH$. A particular vector {\em codeword} $\bc\in\scC$ is transmitted. The received word is $\br=\bc+\bx,$  where $\bx$ is assumed to be a low-rank ``error'' vector. (By {\em rank of a vector} we mean that  there exists a  fixed basis    of $\bbF_{q^n}$ over $\bbF_q$ and the rank of a vector  $\ba\in\bbF_{q^n}^n$ is defined as the rank of the matrix $\bA\in\bbF_{q}^{n\times n}$ whose elements are the coefficients of $\ba$ in the basis. See \cite[Sec.~VI.A]{Sil08} for details of this isomorphic map.)   The optimization problem for decoding $\bc$ given $\br$ is then 
\begin{align}
&\minimize \,\quad  \rank(\br-\bc)   \nn\\
&\st \quad\bc\in\scC \label{eqn:rankdec}
\end{align}
which is  identical to the min-rank problem in~\eqref{eqn:minrank} with the identification of the low error vector  $\bx\equiv \br-\bc$. Note that the matrix version of the vector $\br$ (assuming a fixed basis), denoted as $\bR$, satisfies the linear constraints in~\eqref{eqn:linear_meas}.    Since the assignment $(\bA,\bB)\mapsto\rank(\bA-\bB)$ is a metric on the space of matrices~\cite[Sec.~II.B]{Sil08}, the problem in~\eqref{eqn:rankdec} can be interpreted as a minimum (rank) distance decoder. 
\subsection{Distance Properties of Equiprobable    Rank-Metric Codes} \label{sec:equip_code}
We formalize the notion of an equiprobable linear code and analyze its rank distance properties in this section. The results we derive here are the rank-metric analogs of the results in Barg and Forney~\cite{Barg02} and will prove to be useful in shedding light on the  geometry involved in the sufficient condition  for recovering  the  unknown    low-rank matrix $\bX$ in Proposition~\ref{prop:uniform}.
\begin{definition}
A {\em rank-metric code} is a non-empty subset of $\bbF_q^{n\times n}$ endowed with the the rank distance  $(\bA,\bB)\mapsto\rank(\bA-\bB)$. 
\end{definition}
\begin{definition}
We say that $\scC\subset \bbF_q^{n\times n}$ is an {\em equiprobable linear   rank-metric code} if 
\begin{equation}
\scC:=\{ \bC\in\bbF_q^{n\times n}: \lrangle{\bC}{\rvbH_a}=0 , a\in [k]\} \label{eqn:random_code}
\end{equation}
where $\rvbH_a,a\in [k]$ are random matrices where each entry is statistically independent of other entries and equiprobable in $\bbF_q$, i.e., with pmf given in~\eqref{eqn:uar}.  Each matrix    $\bC\in\scC$ is called a {\em codeword}. Each matrix $\rvbH_a$ is said to be a {\em parity-check matrix}.
\end{definition}
Recall that the inner product   is defined as  $\lrangle{\bC}{ \rvbH_a}= \mathrm{Tr}(\bC  \,  \rvbH_a^T)$. We reiterate that in  the   coding theory literature~\cite{Gab85, Roth,Sil08, Loi06,  Mon07 , Gad08}, rank-metric codes  usually consist of length-$n$  {\em vectors} $\bc \in\scC$ whose elements belong to the {\em extension field} $\bbF_{q^n}$. We refrain from adopting this approach here as we would like to make direct comparisons to the rank minimization problem, where the measurements are generated as in~\eqref{eqn:linear_meas}.\footnote{The usual approach to defining linear rank-metric codes \cite{Gab85, Roth} is the
following: Every codeword in the codebook, ${\bf c} \in
\mathbb{F}_{q^N}^n$, is required to satisfy the $m$ parity-check
constraints $\sum_{i=1}^n \rvbh_{a,i} {\bf c}_i = 0 \in
\mathbb{F}_{q^N}$ for $a \in [m]$ and where $\rvbh_{a,i} \in
\mathbb{F}_{q^N}$ and ${\bf c}_i \in \mathbb{F}_{q^N}$ are,
respectively, the $i$-th elements of $\rvbh_a$ and ${\bf c}$.  Note
that in the paper we focus on the case $N = n$, but make the
distinction here to connect directly with the coding literature.  We
can reexpress each of these $m$ constraints as $N$ matrix trace constraints
in $\mathbb{F}_q$, per~\eqref{eqn:random_code}, as follows.  Consider any basis
$\mathcal{B}$ for $\mathbb{F}_{q^N}$ over $\mathbb{F}_q$, $\mathcal{B}
= \{\bb_1, \ldots, \bb_{N}\}$, where $\bb_j\in\bbF_{q^N}$.  We represent $\rvbh_{a,i}$ and ${\bf
  c}_i$ in this basis as $\rvbh_{a,i} = \sum_{j=1}^{N}\rvh_{a,i,j}
\bb_j$ and ${\bf c}_{i} = \sum_{k=1}^{N} c_{i,k} \bb_k$, respectively.
Let $\tilde{\rvbH}_a$ be the $n \times N$ matrix whose $(i,j)$-th
entry is the coefficient $\rvh_{a,i,j} \in \mathbb{F}_q$ and ${\bf C}$ be
similarly defined by the $c_{i,k} \in \mathbb{F}_q$.  Now define
$\omega_{j,k,l}$ as the coefficients in $\mathbb{F}_q$ of the
representation of $\bb_j \bb_k$, i.e., $\bb_j \bb_k = \sum_{l=1}^{N}
\omega_{j,k,l} \bb_l$.  Define ${\bf \Omega}_l$ to be the symmetric $N
\times N$ matrix whose $(j,k)$-th entry is $\omega_{j,k,l}$.  By
substituting the expansions for $\rvbh_a$ and ${\bf c}$ into the
standard parity-check definition and making use of the fact that the
basis elements $\bb_j$ are linearly independent, we discover the following:
the constraint $\sum_{i=1}^n \rvbh_{a,i} {\bf c}_i = 0$ is
equivalent to the $N$ constraints $\mathrm{Tr} ({\bf C} {\bf
  \Omega}_l \tilde{\rvbH}_a^T) = 0 \in \mathbb{F}_q$ for $l \in [N]$.
If we define $\tilde{\rvbH}_a {\bf \Omega}_l$ for each $a \in [m]$, $l
\in [N]$ to be one of the constraints in \eqref{eqn:random_code}, we get that the set of
${\bf C}$ matrices $\scC$  satisfying \eqref{eqn:random_code} is the rank-metric codes defined by
the ${\rvbh}_a$, $a \in [m]$.  A simple relation between the ${\bf
  \Omega}_l$ matrices holds if the basis is chosen to be a {\em normal} basis \cite[Def.\ 2.32]{Lidl}.
\iffalse If the basis of $\bbF_{q^n}$ over $\bbF_q$ is {\em normal}~\cite[Def.\ 2.32]{Lidl}, i.e., one that can be expressed as $\calB=\{\beta^{[0]} ,\beta^{[1]},\beta^{[2]},\ldots, \beta^{[n-1]} \}$ for  some $\beta\in\bbF_{q^n}$ and $[i]:=q^i$, then each constraint for   usual linear rank-metric codes \cite{Gab85, Roth} $\langle \rvbh_a, \bc\rangle=0 \in \bbF_{q^n}$   can be rewritten as $\mathrm{Tr}(\bC  \,  (\rvbH_a^{(i)})^T) = 0 \in\bbF_q$ for $i\in [n]$, where the ``parity-check  matrices'' $\rvbH_a^{(i)}, i\in [n]$ can be expressed in terms of the expansion   coefficients of $\rvbh_a$ in the basis $\calB$. Thus, there is a one-to-one correspondence between linear rank-metric codes  as in~\cite{Gab85, Roth}  and that defined in \eqref{eqn:random_code}.  The normal basis theorem \cite[Theorem 2.35]{Lidl} guarantees the existence of a normal basis for any finite Galois extension of fields.     {\em self-dual}~\cite[Sec.~2.3]{Lidl}, then for the constraints in~\eqref{eqn:random_code}, the elements of  the parity-check matrices $\rvbH_a$ are the expansion coefficients of the usual rank-metric parity-check matrix   whose elements belong to the extension field. This establishes a transparent correspondence between the parity-check matrices in the rank-metric  codes literature and the $\rvbH_a$ matrices here. \fi} Hence, the term {\em codewords} will always refer to  {\em matrices} in $\scC$. 
\begin{definition}
The {\em number of codewords}  in the code $\scC$  of rank $r$ ($r=0,1,\ldots, n$) is denoted as $\rvN_{\scC}(r)$.  
\end{definition}
Note that $\rvN_{\scC}(r)$ is a random variable since $\scC\subset \bbF_q^{n\times n}$ is a random subspace. This quantity can also be expressed  as 
\begin{align}
\rvN_{\scC}(r) := \sum_{\bM \in \bbF_q^{n\times n}: \rank(\bM)=r} \bbI\{\bM\in\scC\},  \label{eqn:NCr_def}
\end{align}
where $\bbI\{ \bM\in\scC\}$ is the (indicator) random variable which takes on the value one if $\bM\in\scC$ and zero otherwise. Note that the matrix $\bM$ is deterministic, while the code $\scC$ is random. We remark that the decomposition of   $\rvN_{\scC}(r)$ in~\eqref{eqn:NCr_def}  is different from that in Barg and Forney~\cite[Eq.\ (2.3)]{Barg02}   where the authors considered  and analyzed the analog of the sum
\begin{equation}
\tilde{\rvN}_{\scC}(r) := \sum_{j\in\{1,\ldots,  |\scC|\}\,:\,    \rvbC_j\ne\bzero } \bbI\{ \rank(\rvbC_j)=r\}, \label{eqn:barg_fo}
\end{equation}
where $j\in \{1,\ldots,  |\scC|\}$ indexes the (random) codewords in $\scC$. Note that $\tilde{\rvN}_{\scC}(r)={\rvN}_{\scC}(r)$ for all $r\ge 1$ but they differ when $r=0$ ($\tilde{\rvN}_{\scC}(0)=0$ while ${\rvN}_{\scC}(0)=1$). It turns out that the sum in~\eqref{eqn:NCr_def} is more amenable to analysis given that our parity-check (sensing) matrices $\rvbH_a,a\in [k],$ are random  (as  in Gallager's work  in~\cite[Theorem 2.1]{Gall}) whereas in~\cite[Sec.~II.C]{Barg02}, the generators are random.\footnote{Indeed, if the generators are random, it is easier to derive the statistics of  the number of codewords of rank $r$ using~\eqref{eqn:barg_fo} instead of~\eqref{eqn:NCr_def}. } Recall the rank-dimension ratio $\gamma$ is the limit of the ratio $r/n$ as $n\to\infty$. Using~\eqref{eqn:NCr_def}, we can show the following:
\begin{lemma}[Moments of $ \rvN_{\scC}(r)$] \label{lem:stats}
For $r=0$, $\rvN_{\scC}(r)=1$. For $1\le r\le n$, the mean of $\rvN_{\scC}(r)$ satisfies
\begin{align}
q^{-k+2 rn- r^2  -2r} \le    \bbE \rvN_{\scC}(r) \le 4  q^{-k+2 rn- r^2   }. \label{eqn:expect}
\end{align}
Furthermore,  the variance of $\rvN_{\scC}(r)$ satisfies
\begin{align}
\var(\rvN_{\scC}(r)) &\le \bbE\rvN_{\scC}(r) \label{eqn:var}.
\end{align}
\end{lemma}
The proof of Lemma~\ref{lem:stats} is provided in Appendix \ref{app:stats}.  Observe  from~\eqref{eqn:expect} that the average number of codewords with rank $r$, namely $ \bbE \rvN_{\scC}(r)$, is exponentially large (in $n^2$) if $k< (2-\veps) \gamma (1-\gamma/2)n^2$  (compare to  the converse in  Proposition~\ref{prop:converse}) and exponentially small  if $k> (2+\veps) \gamma (1-\gamma/2)n^2$ (compare to the  achievability in Proposition~\ref{prop:uniform}).  By Chebyshev's inequality, an immediate corollary of Lemma~\ref{lem:stats} is the following:
\begin{corollary}[Concentration of number of codewords of rank $r$]
Let $f_n$ be any sequence such that  $\lim_{n\to\infty}f_n=\infty$. Then,
\begin{equation}
\lim_{n\to\infty}\bbP \left(|\rvN_{\scC}(r)-\bbE \rvN_{\scC}(r)|\ge f_n\sqrt{\bbE \rvN_{\scC}(r)} \right) = 0. \label{eqn:cheb}
\end{equation}
\end{corollary}
Thus, $\rvN_{\scC}(r)$ concentrates to its mean in the sense of~\eqref{eqn:cheb}. A similar result for the random generator case was developed in~\cite[Corollary~1]{Loi06}. Also, our derivations based on Lemma~\ref{lem:stats} are cleaner and require fewer assumptions.  We now define the notion of the  minimum rank distance of a rank-metric code. 
\begin{definition}
The {\em minimum rank distance} of a  rank-metric code $\scC$ is defined as 
\begin{equation} 
\rvd_{\mathrm{R}}(\scC) := \min_{\bC_1,\bC_2\in\scC:\bC_1\ne\bC_2}\rank(\bC_1-\bC_2). \label{eqn:min_dist}
\end{equation}
\end{definition}
By linearity of the code $\scC$, it can be seen that the minimum rank distance in~\eqref{eqn:min_dist} can also be written as
\begin{equation}
\rvd_{\mathrm{R}}(\scC) := \min_{\bC \in\scC:\bC\ne \bzero}\rank(\bC). \label{eqn:min_dist_alt}
\end{equation}
Thus, the minimum rank distance of a linear code is equal to the minimum rank over all non-zero matrix codewords. 
\begin{definition}
The {\em relative minimum rank distance} of a code $\scC\subset\bbF_q^{n\times n}$ is defined as $\rvd_{\mathrm{R}}(\scC) /n$.  
\end{definition}
Note that the relative minimum rank distance is a random variable taking on values in the unit interval. In this section, we assume there exists some $\alpha \in (0,1)$ such that $k /n^2\to\alpha$ (cf.\ Section~\ref{sec:rand}). This is the scaling regime of interest. 


\begin{proposition}[Asymptotic linear independence] \label{prop:li} 
Assume that each random matrix  $\rvbH_a \in\bbF_q^{n\times n}$ consists of   independent entries that are drawn according to the pmf in~\eqref{eqn:prob_delta}. Let $\rvm:=\dim(\mathrm{span} \{\vect(\rvbH_1),\ldots,\vect(\rvbH_k)\})$. If $\delta\in\Omega (\frac{\log n}{n})$,  then $\rvm/k\to 1$ almost surely (a.s.). 
\end{proposition}
The proof of this proposition is a   consequence of a result by Bl\"{o}mer et al.~\cite{Blomer}. We provide the details in Appendix \ref{app:li}. 

We would now like to define the notion of the {\em rate} of a random code.  Strictly speaking, since  $\scC$ is a random linear code, the rate of the code should be defined as the random variable $\tilde{\rvR}_n:= 1-\rvm/n^2$. However,  a consequence of  Proposition~\ref{prop:li} is that   $\tilde{\rvR}_n/(1-k/n^2)\to 1$ a.s.\ if $\delta\in\Omega (\frac{\log n}{n})$. Note that this prescribed rate of decay of $\delta$  subsumes the equiprobable model (of interest in this section) as a special case. (Take $\delta=(q-1)/q$ to be  constant.) In light of Proposition~\ref{prop:li}, we adopt the following definition:
\begin{definition} \label{def:def_rate}
The {\em rate} of the     linear rank-metric  code $\scC$ [as in \eqref{eqn:random_code}] is   defined as 
\begin{equation} 
R_n:= \frac{n^2-k}{n^2}=1-\frac{k}{n^2}. \label{eqn:rate_code}
\end{equation}
\end{definition}
The limit of $R_n$ in \eqref{eqn:rate_code} is denoted as $R \in  [0,1]$. Note also that $\tilde{\rvR}_n/R\to 1$ a.s.

\begin{proposition}[Lower bound on relative minimum distance] \label{prop:small}
Fix $\veps>0$. For any $R\in [0,1]$, the probability that the equiprobable linear code in~\eqref{eqn:random_code} has relative minimum rank distance less than $1-\sqrt{R}-\veps$ goes to zero   as $n\to\infty$.
\end{proposition}
\begin{proof}
Assume\footnote{The restriction that $\veps < 2(1-\gamma)$   is not a serious one since the validity of  the claim in Proposition~\ref{prop:small}  for some $\veps_0>0$ implies the same for all $\veps>\veps_0$. }  $\veps\in (0,2(1-\gamma))$ and define the positive constant $\veps' := 2\veps(1-\gamma)-\veps^2$.  Consider a sequence of ranks $r$ such that $r/n\to\gamma\le 1-\sqrt{R}-\veps$. Fix  $\eta=\veps'/2>0$. Then, by Markov's inequality and~\eqref{eqn:expect}, we have 
\begin{align}
\bbP(\rvN_{\scC}(r)\ge 1) \le \bbE \rvN_{\scC}(r) \le 4 q^{- n^2 [\frac{k}{n^2} -2\gamma(1-\gamma/2)  -\eta]  } ,  \label{eqn:markov}
\end{align}
for all $n>N_{\veps'}$. Since   $\gamma \le 1-\sqrt{R}-\veps$, we may assert by invoking the definition of $R$ that $k\ge (2\gamma(1-\gamma/2)+\veps')n^2$. Hence,  the exponent in square parentheses in~\eqref{eqn:markov} is no smaller than $\veps'/2$. This implies that $\bbP(\rvN_{\scC}(r)\ge 1) \to 0$ or equivalently, $\bbP(\rvN_{\scC}(r)=0)\to 1$. In other words,   there are no  matrices of rank $r$ in the equiprobable linear code $\scC$  with probability at least $1- 4q^{-\veps' n^2/2}$ for all $n>N_{\veps'}$.
\end{proof}
We now introduce some additional notation. We say that two positive sequences $\{a_n\}_{n\in\bbN}$ and $ \{b_n\}_{n\in\bbN}$ are  {\em equal to second order in the exponent} (denoted $a_n\dotdoteq b_n$) if 
\begin{equation}
\lim_{n\to\infty} \frac{1}{n^2} \log_q\frac{a_n}{b_n}=0.
\end{equation}
\begin{proposition}[Concentration of relative minimum distance] \label{prop:large}
Fix $\veps>0$. For any $R\in [0,1]$,  if $r$  is a sequence of ranks  such that  $r/n\to\gamma\ge 1-\sqrt{R}+\veps$, then the probability that  $\rvN_{\scC}(r) \dotdoteq q^{-k+2\gamma(1-\gamma/2) n^2 }$ goes to one as $n\to\infty$. 
\end{proposition}
\begin{proof}
If the sequence of ranks $r$ is such that $r/n\to\gamma\ge 1-\sqrt{R}+\veps$, then the average number of matrices  in the code of rank $r$, namely $\bbE \rvN_{\scC}(r)$,  is exponentially large. By Markov's inequality  and the triangle inequality,
\begin{align} \label{eqn:triangle_ineq}
\bbP(|\rvN_{\scC} (r) - \bbE  \rvN_{\scC} (r) |\ge t)&\le  \frac{\bbE|\rvN_{\scC} (r) - \bbE  \rvN_{\scC} (r) |\ }{t} \nn \\
&\le \frac{2 \bbE  \rvN_{\scC} (r) }{t}.
\end{align}
Choose $t:=q^{ -k+(2\gamma(1-\gamma/2)+\eta)n^2+ n}$, where $\eta$ is given in the proof of Proposition~\ref{prop:small}. Then, applying~\eqref{eqn:expect}  to~\eqref{eqn:triangle_ineq} yields
\begin{align}
\bbP(|\rvN_{\scC} (r) - \bbE  \rvN_{\scC} (r) |\ge t)\le  8 q^{- n}\to 0.
\end{align}
Hence, $\rvN_{\scC} (r) \in ( \bbE  \rvN_{\scC} (r)- t, \bbE  \rvN_{\scC} (r)+t)$ with probability exceeding $1-8 q^{- n}$.    Furthermore, it is easy to verify that $\bbE  \rvN_{\scC} (r)\pm t\dotdoteq  q^{-k+2\gamma(1-\gamma/2) n^2 }$, as desired.
\end{proof}
Propositions~\ref{prop:small} and~\ref{prop:large} allow us to conclude that with probability approaching one (exponentially  fast) as $n\to\infty$, the relative minimum rank distance of the equiprobable linear code in~\eqref{eqn:random_code} is contained in the interval $(1-\sqrt{R}-\veps, 1-\sqrt{R}+\veps)$ for all $R\in [0,1]$. The analog of the  Gilbert-Varshamov   (GV) distance~\cite[Sec.~II.C]{Barg02} is thus
\begin{equation}
\gamma_{\mathrm{GV}}(R):=1-\sqrt{R}.
\end{equation}
Indeed, by substituting the definition of $R$ into $\rvN_{\scC}(r)$ in Proposition~\ref{prop:large}, we see   that a {\em typical} (in the sense of~\cite{Barg02}) equiprobable linear rank-metric code has distance distribution:
\begin{equation}
\rvN_{\mathrm{typ}} (r) \left\{\begin{array}{cc}\!
\dotdoteq \, q^{n^2[R-(1-\gamma)^2]}  &  \gamma\ge \gamma_{\mathrm{GV}}(R)+\veps ,\\
\!\!\!\!\!\!\!\!\!\!\!\!\!\!\!\!\!\!\!\!\!\!\!\!\!\!\!\!\!\!\! =\, 0 & \gamma\le\gamma_{\mathrm{GV}}(R)-\veps. \label{eqn:Ntyp}
\end{array} \right.
\end{equation}
We again remark that Loidreau in~\cite[Sec.~5]{Loi06} also derived results for uniformly random  linear codes in the rank-metric that are somewhat similar to  Propositions~\ref{prop:small} and~\ref{prop:large}. However, our derivations are   more straightforward and require fewer assumptions. As mentioned above, we   assume that the parity-check matrices $\rvbH_a,a\in [k],$ are random  (akin to~\cite[Theorem~2.1]{Gall}), while the assumption in~\cite[Sec.~5]{Loi06} is that the {\em generators} are random  {\em and} linearly independent. Furthermore, to the best of our knowledge, there are no previous studies on the minimum distance properties for  the sparse parity-check matrix setting.  We do this in Section~\ref{sec:sparse_code}.

From the rank distance properties, we can re-derive the achievability (weak recovery) result in Proposition~\ref{prop:uniform} by using the definition of $R$   and solving the following inequality for~$k$:
\begin{equation}
1-\sqrt{R} -\veps\ge \gamma .
\end{equation}
This provides geometric intuition as to why the min-rank decoder succeeds {\em on average}; the typical  relative minimum rank distance of the code should exceed the rank-dimension ratio for successful error correction. We derive a stronger condition (known as the strong recovery condition) in Section~\ref{sec:str_ach}.
\subsection{Distance Properties of Sparse   Rank-Metric Codes}\label{sec:sparse_code}
In this section, we derive the analog of Proposition~\ref{prop:small} for the case where the code $\scC$ is characterized by sparse sensing  (or measurement or parity-check) matrices $\rvbH_a,a\in [k]$. 
\begin{definition} We say that $\scC$ is a {\em $\delta$-sparse linear  rank-metric code} if $\scC$ is as in~\eqref{eqn:random_code} and   where $\rvbH_a,a\in [k]$ are random  matrices where each entry is statistically independent and drawn from the pmf $P_{\rvh}(\fndot;\delta,q)$  defined in~\eqref{eqn:prob_delta}.
\end{definition}
To analyze the number of matrices of rank $r$ in this random ensemble $\rvN_{\scC}(r)$, we partition  the sum in~\eqref{eqn:NCr_def} into subsets of matrices based on their Hamming weight, i.e., 
\begin{equation}
\rvN_{\scC}(r) = \sum_{d=0}^{n^2} \,\, \sum_{ \bM \in \bbF_q^{n\times n}: \rank(\bM)=r  , \|\bM\|_0=d} \bbI\{\bM\in\scC\}  . \label{eqn:NCsparse}
\end{equation}
Define $\theta(d;\delta,q,k) := [q^{-1}+(1-q^{-1}) (1-\delta/(1-q^{-1}))^d]^k$. As shown in Lemma~\ref{lem:circconv} in Appendix \ref{app:sparse}, this is the probability that a  non-zero matrix $\bM$ of Hamming weight $d$ belongs to the $\delta$-sparse code $\scC$. We can demonstrate the following important bound for the  $\delta$-sparse linear   rank-metric code:
\begin{lemma}[Mean of $\rvN_{\scC}(r)$ for sparse codes] \label{lem:mean_sparse}
For $r=0$, $\rvN_{\scC}(r)=1$. If $1\le r \le n$ and $\eta>0$, 
\begin{align}
\bbE \rvN_{\scC}(r) &\le    2^{n^2 \Hb(\beta) }(q-1)^{\beta n^2} (1-\delta)^k + \nn\\*
&\,\,\,+ 4 n^2 q^{n^2  \left[2\gamma (1-\gamma/2) +\eta + \frac{1}{n^2} \log_q \theta(\ceil{\beta n^2};\, \delta,q,k)  \right]}, \label{eqn:mean_sparse} 
\end{align}
for all $\beta\in [0,1/2]$ and all $n\ge N_{\eta}$.
\end{lemma}
By using the sum in~\eqref{eqn:NCsparse}, one sees that this lemma can be justified in exactly the same way as Theorem~\ref{prop:sparse} (See steps leading to~\eqref{eqn:trivialplus} and~\eqref{eqn:trivialplus2} in Appendix~\ref{app:sparse}). Hence, we omit its proof. Lemma~\ref{lem:mean_sparse} allows us to find a tight upper bound on the expectation of $ \rvN_{\scC}(r)$ for the sparse linear rank-metric code by optimizing over the free parameter $\beta\in [0,1/2]$. It turns out $\beta=\Theta(\frac{\delta}{\log n})$ is optimum. In analogy to Proposition~\ref{prop:small} for the equiprobable linear rank-metric code, we can demonstrate  the following for the sparse linear rank-metric   code. 
\begin{proposition}[Lower bound on relative minimum distance for sparse codes] \label{prop:small_sparse}
Fix $\veps>0$ assume that   $\delta = \Omega(\frac{\log n}{n}) \cap o(1)$. For any $R\in [0,1]$, the probability that the   sparse linear code   has relative minimum distance less than $1-\sqrt{R}-\veps$ goes to zero   as $n\to\infty$.
\end{proposition}
\begin{proof}
The condition on the minimum distance implies that $k> (2+\tilde{\veps})\gamma (1-\gamma/2)n^2$ for  some $\tilde{\veps} >0$ (for sufficiently small $\veps$). See detailed argument in proof of Proposition~\ref{prop:small}. This implies from   Theorem~\ref{prop:sparse}, Lemma~\ref{lem:mean_sparse} and  Markov's inequality   that $\bbP(\rvN_{\scC}(r)\ge 1)\to 0$.
\end{proof}
Proposition~\ref{prop:small_sparse} asserts that the relative minimum rank distance  of a $\delta=\Omega(\frac{\log n}{n})$-sparse linear rank-metric  code is  at least $1-\sqrt{R}-\veps$ w.h.p. Remarkably, this property is exactly the  same as that of a (dense) linear   code  (cf.\ Proposition~\ref{prop:small}) in which  the entries in the parity-check matrices $\rvbH_a$ are  statistically independent and equiprobable in $\bbF_q$.   The fact that the (lower bounds on the) minimum distances of both ensembles of codes coincide explains why the min-rank decoder matches the information-theoretic lower bound (Proposition~\ref{prop:converse}) in   the sparse  setting (Theorem~\ref{prop:sparse}) just as in the dense one (Proposition~\ref{prop:uniform}). Note that only an upper bound  of  $\bbE \rvN_{\scC}(r)$  as in~\eqref{eqn:mean_sparse}  is required to make this claim.
\subsection{Strong Recovery} \label{sec:str_ach}
We now utilize the insights gleaned from this section to derive results for  strong recovery (See Section~\ref{sec:weakvsstr} and also~\cite[Sec.~2]{Eldar11} for definitions) of low-rank matrices from linear measurements. Recall that in strong recovery, we are interested in recovering {\em all} matrices whose ranks are no larger than $r$. We contrast this to weak recovery where   a matrix $\bX$ (of low rank)  is fixed and we ask how many random measurements are needed to estimate $\bX$ reliably. 

\begin{proposition}[Strong recovery for uniform measurement model] \label{lem:strong_ach}
Fix $\veps>0$. Under the uniform measurement model, the min-rank decoder recovers {\em all} matrices of rank less than or equal to $r$    with probability approaching one as $n\to\infty$ if 
\begin{equation}
k>(4+\veps)\gamma(1-\gamma)n^2 . \label{eqn:4times}
\end{equation} 
\end{proposition}
We contrast this to the weak achievability result (Proposition~\ref{prop:uniform}) in which $\bX$ with $\rank(\bX)\le r$ was {\em fixed} and we showed that if $k>(2+\veps)\gamma(1-\gamma/2)n^2$, the min-rank decoder recovers $\bX$ w.h.p. Thus, Proposition~\ref{lem:strong_ach} says that if $\gamma$ is small,   roughly twice as many measurements  are needed for strong recovery vis-\`{a}-vis weak recovery. These fundamental limits (and the increase in a factor of 2 for strong recovery) are exactly analogous those developed by Draper and Malekpour in~\cite{Dra09} in the context of compressed sensing over finite fields and  Eldar et al.~\cite{Eldar11} for the problem of rank minimization over the reals.   Given our derivations in the preceding subsections, the proof of this result is   straightforward.
\begin{proof}
We showed in Proposition~\ref{prop:small} that with probability approaching one (exponentially fast), the relative minimum distance of $\scC$ is no smaller than $1-\sqrt{R}-\tilde{\veps}$ for any $\tilde{\veps}>0$. As such to guarantee strong recovery, we need the decoding regions (associated to each codeword in $\scC$) to be disjoint. In other words, the rank distance between any two distinct codewords $\bC_1,\bC_2\in\scC$ must exceed $2r$. See Fig.~\ref{fig:disjoint} for an illustration. In terms of the relative minimum rank distance $1-\sqrt{R} -\tilde{\veps}$, this requirement translates to\footnote{The strong recovery requirement in \eqref{eqn:str_criterion} is analogous to the well-known fact that in the binary Hamming case, in order to correct {\em any} vector $\br=\bc+\be$ corrupted with $t$ errors (i.e., $\|\be\|_0 = t$) using minimum distance decoding, we {\em must} use a code with minimum distance at least $2t+1$. \label{fn:str}}
\begin{align}
1-\sqrt{R} -\tilde{\veps} \ge  2\gamma. \label{eqn:str_criterion}
\end{align}
Rearranging this inequality  as and using the definition of $R$ [limit of $R_n$ in \eqref{eqn:rate_code}] as we did in Proposition~\ref{prop:small}  yields the required  number of measurements prescribed.  
\end{proof}

In analogy to Proposition~\ref{lem:strong_ach}, we can show the following for the sparse model. 

\begin{proposition}[Strong recovery  for sparse measurement model]\label{lem:strong_ach2}
Fix $\veps>0$. Under the $\delta=\Omega(\frac{\log n}{n})$-sparse measurement model, the min-rank decoder recovers {\em all} matrices of rank less than or equal to $r$ with probability approaching one as $n\to\infty$  if  \eqref{eqn:4times} holds.
\end{proposition}
\begin{proof}
The proof uses Proposition~\ref{prop:small_sparse} and follows along the exact same lines as that of Proposition~\ref{lem:strong_ach}. 
\end{proof}

\begin{figure}
\centering
\includegraphics[width = 1.7in]{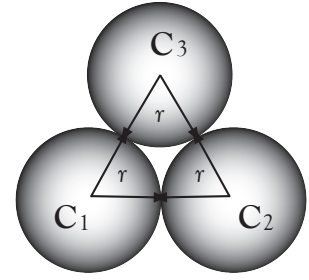}
\caption{For strong recovery, the decoding regions associated to each codeword $\bC\in\scC$ have to be disjoint, resulting in the criterion in~\eqref{eqn:str_criterion}.}
\label{fig:disjoint}
\end{figure}

%% file: nonexhaust_v2.tex
\section{Reduction in the Complexity of the Min-Rank Decoder} \label{sec:nonexhaust}
In this section, we devise a procedure to reduce the  complexity for min-rank decoding  (vis-\`{a}-vis exhaustive search). This  procedure is inspired by techniques in the cryptography literature~\cite{Cha96, Our02}. We adapt the techniques for our problem which is somewhat different. As we   mentioned in Section~\ref{sec:dist}, the codewords in this paper are matrices  rather than vectors  whose elements belong to an extension field~\cite{Cha96, Our02}.    

Recall that in min-rank decoding~\eqref{eqn:minrank},   we search for a matrix $\bX\in\bbF_q^{N\times n}$ of minimum rank that satisfies the linear constraints. In this section,  for clarity of exposition, we differentiate between the number of rows ($N$) and the number of columns ($n$) in $\bX$. The vector $\rvy^k$ is known as the   {\em syndrome}.  

We first suppose that the minimum  rank in~\eqref{eqn:minrank} is known to be  equal to  some integer $r\le\min\{N,n\}$. Since our proposed algorithm requires exponentially many elementary operations  (addition and multiplication) in $\bbF_q$, this assumption does not   affect the time complexity significantly. Then the problem in~\eqref{eqn:minrank} reduces to a satisfiability problem: Given an integer $r$, a collection of parity-check matrices $\rvbH_a,a\in [k]$ and a syndrome vector  $\rvy^k$, find (if possible) a matrix $\bX\in \bbF_q^{N\times n}$ of rank exactly equal to $r$ that satisfies the linear constraints  in~\eqref{eqn:minrank}.  Note that the constrains in~\eqref{eqn:minrank} are equivalent to $\lrangle{\vect(\rvbH_a)}{\vect(\bX)}=\rvy_a,a\in [k]$.

We first claim that we can, without loss of generality, assume that $\rvy^k=0^k$, i.e,  the constraints in~\eqref{eqn:minrank} read
\begin{equation}
\langle\rvbH_a,\bX\rangle=0,\quad a\in[k]. \label{eqn:zero_eq}
\end{equation}
We justify this claim as follows: Consider the new  syndrome-augmented vectors $[\vect(\rvbH_a) ;\rvy_a]^T\in\bbF_q^{Nn+1}$ for every $a\in [k]$. Then, every solution $\vect(\bX')$ of the system of equations  
\begin{equation}
\lrangle{ [\vect(\rvbH_a) ;\rvy_a]}{ \vect(\bX')}=0,\quad a\in[k]  \label{eqn:larger_constr}
\end{equation}
can be partitioned into two parts, $\vect(\bX') = [\vect(\bX_1); x_2]$ where $\vect(\bX_1)\in\bbF_q^{Nn}$ and $x_2\in\bbF_q$. Thus, every solution of~\eqref{eqn:larger_constr} satisfies  one of two conditions:
\begin{itemize}
\item $x_2=0$. In this case $\bX_1$ is a solution to the linear equations in~\eqref{eqn:minrank}. 
\item  $x_2\ne 0$. In this case $\bX_1$ solves  $\langle\rvbH_a,\bX_1\rangle=x_2\rvy_a$. Thus, $x_2^{-1} \bX_1$ solves \eqref{eqn:minrank}.
\end{itemize}
This is also known as {\em coset decoding}. Now, observe that since it is known that $\bX$ has rank equal to $r$ (which is assumed known), it can be written as 
\begin{equation}
\bX = \sum_{l=1}^r \bu_l \bv_l^T = \bU \bV^T
\end{equation}
where each of the vectors $\bu_l\in\bbF_q^N$ and  $\bv_l \in \bbF_q^n$. The matrices $\bU\in\bbF_q^{N\times r}$ and $\bV \in\bbF_q^{n\times r}$ are of (full) rank $r$ and are referred to as the {\em basis matrix} and the {\em coefficient matrix} respectively. The linear system of equations in~\eqref{eqn:zero_eq} can be expanded as 
\begin{equation}
\sum_{l=1}^r\,  \sum_{i=1}^N \,\sum_{j=1}^n \, [\rvbH_a]_{i,j}\,  u_{l,i} \, v_{l,j}=0, \qquad a\in [k]  \label{eqn:quadr}
\end{equation}
where $\bu_l = [u_{l,1},\ldots, u_{l,N}]^T$ and  $\bv_l = [v_{l,1},\ldots, v_{l,n}]^T$. Thus, we need to solve a system of {\em quadratic equations} in the basis elements $u_{l,i}$ and the coefficients $v_{l,j}$. 

\subsection{Na\"{i}ve Implementation}
A na\"{i}ve way to find a consistent $\bU$ and $\bV$ for~\eqref{eqn:quadr} is to employ the following algorithm: 
\begin{enumerate} 
\item Start with  $r=1$. 
\item Enumerate all   bases $\bU=\{u_{l,i}: i\in [N], l\in [r]\}$. 
\item For each basis,   solve (if possible) the resulting linear system of equations in $\bV=\{v_{l,j}: j\in [n], l\in [r]\}$. 
\item If a consistent set of coefficients $\bV$ exists (i.e., \eqref{eqn:quadr} is satisfied), terminate and set $\bX=\bU\bV^T$.  Else increment $r\leftarrow r+1$ and go to step 2.  
\end{enumerate}
The second step can be solved easily if the number of equations is less than or equal to the number of unknowns, i.e., if $nr\ge k$.  However, this is usually not the case since for successful recovery, $k$ has to satisfy \eqref{eqn:ach} so, in general, there are more equations (linear constraints) than unknowns. We attempt to solve for (if possible) a consistent $\bV$, otherwise we increment the guessed rank $r$. The computational complexity of this  na\"{i}ve  approach (assuming $r$ is known and so no iterations over $r$ are needed) is $O( (nr)^3 q^{Nr})$ since there are $q^{Nr}$ distinct bases and solving the linear system via Gaussian elimination requires at most $O((nr)^3)$    operations in $\bbF_q$. 

\subsection{Simple Observations to Reduce the Search for the   Basis $\bU$}
We now use ideas from~\cite{Cha96, Our02} and make two simple observations to dramatically reduce the search for the   basis in   step 2 of the above na\"{i}ve implementation. 

\emph{Observation (A)}: Note that if $\tilbX$ solves~\eqref{eqn:zero_eq}, so does $\rho\tilbX$ for any $\rho\in\bbF_q$. Hence, without loss of generality, we may assume that the we can scale the (1,1) element of $\bU$ to be equal to 1.  The number of bases we need to enumerate may  thus be reduced by  a factor of $q$. 

\emph{Observation (B)}:  Note that the decomposition $\bX = \bU \bV^T$ is not unique. Indeed if $\bX = \bU \bV^T$, we may also decompose $\bX$ as  $\bX = \tilbU \tilbV^T$, where $\tilbU  = \bU\bT$ and $\tilbV =\bV \bT^{-T}$ and $\bT$ is {\em any} invertible $r\times r$ matrix over $\bbF_q$. We say that two bases $\bU,\tilbU$ are {\em equivalent}, denoted $\bU \sim \tilbU$, if there exists an invertible matrix $\bT$ such that $\bU = \tilbU \bT$. The equivalence relation $\sim$ induces a partition of the set of $\bbF_q^{N\times r}$ matrices. 

Let $[\bU] := \{\tilbU  \in\bbF_q^{N\times r}: \tilbU \sim \bU\}$ be the equivalence class of matrices containing the matrix $\bU$.  From the preceding discussion on the indeterminacies in the decomposition of the low rank matrix $\bX$, we observe  that the complexity involved in the enumeration of all $\bbF_q^{N\times r}$ matrices in step 2 in the na\"{i}ve implementation can be   reduced by only enumerating  the different equivalence classes induced by $\sim$. More precisely, we find (if possible) coefficients $\bV$ for a  basis $\bU$ from each equivalence class, e.g.,  $\bU_1 \in [\bU_1],\ldots, \bU_m \in [\bU_m]$. Note that  the number of equivalence classes (by Lagrange's theorem) is
\begin{align}
m = \frac{q^{Nr}}{\Phi_q(r,r)}\le   4  q^{r(N-r)} , \label{eqn:4factor}
\end{align}
where recall from Section~\ref{sec:bounds_number} that $\Phi_q(r,r)$ is the number of non-singular   matrices in $\bbF_q^{r\times r}$.  The inequality  arises from the fact that $\Phi_q(r,r)\ge\frac{1}{4} q^{r^2}$, a simple  consequence of~\cite[Cor.~4]{Cha96}.  Algorithmically, we can enumerate  the equivalence classes  by first considering all matrices of the form 
\begin{equation} \label{eqn:Uconstruct}
\bU = \begin{bmatrix} 
 \bI_{r\times r} \\ \bQ
\end{bmatrix}  ,
\end{equation}
where  $\bI_{r\times r}$ is the identity matrix of size $r$, and $\bQ$ takes on all possible values in $\bbF_q^{(N-r)\times r}$.  Note that if $\bQ$ and $\tilbQ$ are distinct, the corresponding $\bU = [\bI ; \bQ^T]^T$ and $\tilbU=[\bI ; \tilbQ^T]^T$  belong to different equivalence classes. However, the top $r$ rows of $\bU$ may not be linearly independent so we have yet to consider all equivalence classes. Hence, we  subsequently  permute the rows of  each previously considered $\bU$  to ensure every equivalence class is considered.

From the considerations in (A) and (B), the computational complexity can be reduced from $O(  (nr)^3 q^{Nr})$ to $O( (nr)^3 q^{ r(N-r)-1})$. By further noting that there is symmetry between the basis matrix  $\bU$ and the coefficient matrix $\bV$, we see that the resulting computational complexity is $O( (\max\{n,N\}r)^3 q^{ r(\min\{n,N\}-r)-1})$. Finally, to  incorporate the fact that $r$ is unknown, we start the procedure assuming $r=1$, proceed to $r\leftarrow r+1$ if there does not exist a consistent solution and so on,  until a consistent $(\bU,\bV)$ pair is found. The resulting computational complexity is  thus
$
O( r(\max\{n,N\}r)^3 q^{ r(\min\{n,N\}-r)-1}).
$

%% file: conclusion_v2.tex
\section{Discussion and Conclusion} \label{sec:concl}
In this section, we  elaborate on connections of our work to   the related works mentioned the introduction and in Tables~\ref{tab:comp} and~\ref{tab:comp2}. We will also conclude the paper by summarizing our main  contributions and suggesting   avenues for future research.

\subsection{Comparison to existing coding-theoretic techniques for rank minimization over finite fields} \label{sec:coding_conn}
In general, solving the min-rank decoding problem~\eqref{eqn:rankdec} is intractable (NP-hard). However,  it is known that if the linear operator $\bH$ (in~\eqref{eqn:linear_stack} characterizing the code $\scC$) admits a favorable algebraic structure, then one can estimate  a sufficiently low-rank (vector with elements in the extension field $\bbF_{q^n}$ or matrix   with elements in  $\bbF_{q}$) $\bx$  and thus the codeword $\bc$ from the received word $\br$  efficiently (i.e.,  in polynomial time). For instance, the class of {\em Gabidulin codes} \cite{Gab85, Roth}, which are rank-metric analogs of Reed-Solomon codes, not only achieves the Singleton bound and thus has maximum rank distance (MRD), but decoding  can be achieved using a modified form of the Berlekamp-Massey algorithm (See \cite{Rich04} for example). However, the   algebraic structure  of the codes (and in particular the mutual  dependence between the equivalent  $\rvbH_a$ matrices)  does not permit the line of analysis we adopted. Thus
it is unclear how many linear measurements would be required in order
to guarantee recovery using the suggested code structure. Silva, Kschischang and K\"{o}tter  \cite{Sil08} extended the Berlekamp-Massey-based algorithm to handle errors and erasures for the purpose of error control in linear random network coding.  In both these cases, the underlying error matrix is assumed to be deterministic and the algebraic structure on the parity check matrix permitted efficient decoding based on  error locators.

In another related work, Montanari and Urbanke~\cite{Mon07} assumed that the    error matrix $\rvbX$ is drawn uniformly at random from all matrices of {\em known} rank $r$. The authors then constructed a sparse  parity check  code (based on a sparse factor graph). Using an ``error-trapping'' strategy by constraining codewords to have rows  that are have zero Hamming weight without any loss of  rate,  they first learned the rowspace of $\rvbX$ before adopting a (subspace) message passing strategy to complete the reconstruction. However, the  dependence across rows of the parity check matrix (caused by lifting)  violates the  independence assumptions needed for our analyses to hold.  The ideas in~\cite{Mon07} were subsequently extended by Silva, Kschischang and K\"{o}tter~\cite{Sil10} where the authors  computed the information capacity of various (additive and/or multiplicative) matrix-valued channels over finite fields. They also devised  ``error-trapping'' codes to achieve   capacity.   However, unlike this work,   it is assumed in~\cite{Sil10} that the underlying low-rank error matrix is chosen uniformly. As such, their guarantees do not apply to so-called   crisscross error patterns~\cite{Roth97, Rich04} (see Fig.~\ref{fig:crisscross}), which are of interest in data storage applications.

\input{crisscross}

Our work in this paper is focused primarily on understanding  the fundamental limits of rank-metric codes that are {\em random}. More precisely, the codes are characterized by  either {\em dense}  or {\em sparse}  sensing (parity-check) matrices. This is in contrast to the literature on rank-metric codes (except~\cite[Sec.~5]{Loi06}), in which deterministic constructions predominate.  The codes presented in Section~\ref{sec:dist} are random. However,     in   analogy to the  random coding argument for channel coding~\cite[Sec.~7.7]{Cov06},  if the ensemble of random codes  has low average error probability, there exists a deterministic code that has low error probability. In addition, the strong recovery   results in Section~\ref{sec:str_ach} allow us to conclude that our analyses apply {\em to all} low-rank matrices  $\bX$ in both   equiprobable and sparse settings. This completes all remaining entries in Table~\ref{tab:comp2}.   

Yet another line of research on rank minimization over finite fields (in particular over $\bbF_2$) has been conducted by the combinatorial optimization and graph theory  communities. In~\cite[Sec.~6]{Gro81} and~\cite[Sec.~1]{Peeters} for example, it was demonstrated that if the code (or set of linear constraints) is characterized by a perfect graph,\footnote{A {\em perfect graph} $G$ is one in which each  induced subgraph $H\subset G$  has a  chromatic number $\chi(H)$ that is the same as its clique number $\omega(H)$.}  then the rank minimization problem can be solved  exactly and in polynomial time by the ellipsoid method (since the problem can be stated as a semidefinite program).  In fact, the rank minimization problem is also intimately related to Lov\'{a}sz's $\theta$ function~\cite[Theorem~4]{Lov81}, which characterizes the Shannon capacity of  a graph.

\subsection{Conclusion and Future Directions}
In this paper, we derive  information-theoretic limits for recovering a  low-rank matrix with elements over a finite field given noiseless or noisy linear measurements. We show  that even if the random sensing (or parity-check) matrices are very sparse, decoding can be done with {\em exactly} the same number of measurements as when  the sensing matrices are dense. We then adopt  a coding-theoretic   approach and derived minimum rank distance properties of   sparse random rank-metric codes. These results provide   geometric insights as to  how and why decoding succeeds when sufficiently many measurements are available. The work herein could potentially lead to the design of low-complexity sparse codes for rank-metric channels. 

It is also of interest to analyze whether the sparsity
factor of $\Theta(\frac{\log n}{n})$ is the smallest possible and whether there is a
fundamental tradeoff between this sparsity factor and the number of
measurements required for reliable recovery of the low-rank matrix.
Additionally, in many of the applications that motivate this problem,
the sensing matrices fixed by the application and will not be random; take for example
deterministic parity-check matrices that might define a rank-metric
code. In rank minimization in the real field there are properties about the sensing matrices, and about the underlying matrix being estimated, that can be checked (for example the restricted isometry property~\cite[Eq.~(1)]{Meka}, or random
point sampling joint with incoherence of the low-rank matrix) that, if they are satisfied, guarantee that the true   matrix of interest can be recovered using convex programming. It is of interest to identify an
analog in the finite field, that is, a necessary (or sufficient)
condition on the sensing matrices  and the underlying matrix such that recovery is guaranteed. We would like  to develop   tractable algorithms along the lines of those in Table~\ref{tab:comp} or in the work by Baron et al.~\cite{Bar10}  to solve the min-rank optimization  problem approximately  for particular classes of sensing matrices such as the sparse random ensemble.

Finally,  Dimakis and Vontobel~\cite{Dimakis09} make an intriguing connection between linear programming (LP) decoding for channel coding and LP decoding for compressed sensing. They reach known compressed sensing results via a new path --   channel coding. Analogously, we wonder whether known   rank minimization results can be derived using rank-metric coding tools, thereby providing novel interpretations. And just as in~\cite{Dimakis09}, the reverse direction is also open.  That is, whether the growing literature and understanding of rank minimization problems could be leveraged to design more tractable and interesting decoding approaches for rank-metric codes.

%% file: crisscross.tex
\begin{figure}
\centering
\begin{picture}(100,120)
\put(20, 0){\line(0,1){100}}
\put(30, 0){\line(0,1){100}}

\put(70, 0){\line(0,1){100}}
\put(80, 0){\line(0,1){100}}

\put(0,40){\line(1,0){100}}
\put(0,50){\line(1,0){100}}

\linethickness{0.3mm}
\put(0,0){\line(1,0){100}}
\put(0,0){\line(0,1){100}}
\put(0,100){\line(1,0){100}}
\put(100, 0){\line(0,1){100}}

\put(25, 20){\circle*{3}}
\put(25, 60){\circle*{3}}
\put(25, 80){\circle*{3}}

\put(75, 10){\circle*{3}}
\put(75, 58){\circle*{3}}
\put(75, 75){\circle*{3}}
\put(75, 95){\circle*{3}}

\put(25, 45){\circle*{3}}
\put(75, 45){\circle*{3}}
\put(50, 45){\circle*{3}}
\put(8, 45){\circle*{3}}
\put(92, 45){\circle*{3}}

\put(-10,60){\vector(0,1){40}}
\put(-10,40){\vector(0,-1){40}}
\put(-13, 47){\mbox{$n$}}

\put(40,110){\vector(-1,0){40}}
\put(60,110){\vector(1,0){40}}
\put(48,108){\mbox{$n$}}
\end{picture}
\caption{Probabilistic crisscross error patterns \cite{Roth97}: The figure shows an error matrix $\bX$. The non-zero values (indicated as black dots) are restricted to two columns and one row. Thus, the rank of the error matrix $\bX$ is  at most three.  }
\label{fig:crisscross}
\end{figure}
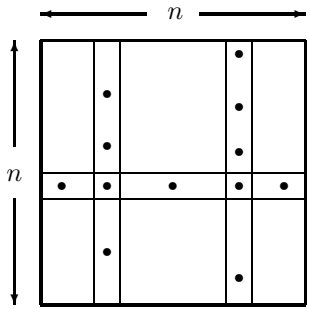

%% file: pairwise_v2.tex
\section{Proof of Lemma~\ref{lem:pairwise} } \label{app:pairwise}
\begin{proof}
It suffices to show that  the conditional probability $\bbP(\calA_{\bZ'}|\calA_{\bZ})=\bbP (\calA_{\bZ'})=q^{-k}$ for $\bZ\ne\bZ'$. We define the non-zero matrices $\bM:=\bX-\bZ$ and $\bM':=\bX-\bZ'$. Let $\calK := \supp(\bM'-\bM)$ and $\calL:=\supp(\bM)$. The idea of the proof is to partition the joint support $\calK\cup\calL$ into disjoint sets. More precisely, consider 
\begin{align}
 \bbP(\calA_{\bZ'}|\calA_{\bZ})&\eqa \bbP(\lrangle{\bM'}{\rvbH_1}=0 \, |\, \lrangle{\bM}{\rvbH_1}=0)^k  \nn\\
&\eqb \bbP(\lrangle{\bM'\!-\!\bM}{\rvbH_1}\!=0\, |\, \lrangle{\bM}{\rvbH_1}=0)^k, \label{eqn:pair1}
\end{align}
where $(a)$ is from the definition of $\calA_{\bZ} := \{\lrangle{\bX-\bZ}{\rvbH_a}=0,\forall\, a\in [k]\}$ and the independence of the random matrices $\rvbH_a,a\in [k]$ and $(b)$ by linearity. It suffices to show that the probability in~\eqref{eqn:pair1} is $q^{-1}$. Indeed, 
\begin{align}
&\bbP(\lrangle{\bM'-\bM}{\rvbH_1}=0\, |\, \lrangle{\bM}{\rvbH_1}=0) \nn\\
&\eqc \bbP \Big(\sum_{(i,j)\in\calK} [\bM'-\bM]_{i,j} [\rvbH_1]_{i,j} =0 \, \Big|\, \nn\\*
      &\qquad\qquad\qquad \sum_{(i,j)\in\calL} [\bM]_{i,j} [\rvbH_1]_{i,j}  =0 \Big)  \nn\\
&\eqd \bbP \Big(\sum_{(i,j)\in\calK}  [\rvbH_1]_{i,j} =0 \,\Big|\, \sum_{(i,j)\in\calL}  [\rvbH_1]_{i,j}  =0 \Big)  \label{eqn:KL} , 
\end{align}
where $(c)$ is from the definition of the  inner product and the sets $\calK$ and $\calL$, $(d)$ from the fact that $[\bM]_{i,j}[\rvbH_1]_{i,j}$ has the same distribution as $[\rvbH_1]_{i,j}$ since $[\bM]_{i,j}\ne 0$ and $[\rvbH_1]_{i,j}$ is uniformly distributed in $\bbF_q$. Now, we split the sets $\calK$ and $\calL$ in~\eqref{eqn:KL} into two disjoint subsets each, obtaining
\begin{align}
&\bbP(\lrangle{\bM'-\bM}{\rvbH_1}=0\, |\, \lrangle{\bM}{\rvbH_1}=0) \nn\\
&= \bbP \Big(\sum_{(i,j)\in\calK\setminus \calL}  [\rvbH_1]_{i,j} +\sum_{(i,j)\in\calL\cap\calK}  [\rvbH_1]_{i,j}=0\, \Big| \,  \nn\\*
   &\qquad\qquad\qquad  \sum_{(i,j)\in\calL\setminus \calK}  [\rvbH_1]_{i,j}  +\sum_{(i,j)\in\calL\cap \calK}  [\rvbH_1]_{i,j}=0 \Big)  \nn\\
&\eqe \bbP \Big( \sum_{(i,j)\in\calL\setminus \calK}  [\rvbH_1]_{i,j} =\sum_{(i,j)\in\calK\setminus \calL}  [\rvbH_1]_{i,j}\, \Big| \, \nn\\*
   &\qquad  \qquad   \sum_{(i,j)\in\calL\setminus \calK}  [\rvbH_1]_{i,j}  =-\sum_{(i,j)\in\calL\cap \calK}  [\rvbH_1]_{i,j}  \Big)  \eqf q^{-1} , \nn
\end{align}
Equality $(e)$ is by using the condition   $\sum_{(i,j)\in\calL\setminus \calK}  [\rvbH_1]_{i,j}  = - \sum_{(i,j)\in\calL\cap \calK}  [\rvbH_1]_{i,j}$ and finally  $(f)$ from the fact that the sets $\calK\setminus\calL$, $\calL\setminus\calK$  and $\calL\cap\calK$ are mutually disjoint  so the probability  is $q^{-1}$ by independence and uniformity of $[\rvbH_1]_{i,j}, (i,j)\in [n]^2$.
\end{proof}

%% file: noisy_v2.tex
\section{Proof of Proposition~\ref{prop:noisy}} \label{app:noisy1}
\begin{proof}
Recall   the   optimization problem   for the noisy case in~\eqref{eqn:minrank2} where the optimization variables are $\tilde{\bX}$ and $\tilde{\bw}$. Let $\calS^{\mathrm{noisy}} \subset \bbF_q^{n\times n} \times\bbF_q^k$ be the set of optimizers. In analogy to~\eqref{eqn:err_event}, we define the  ``noisy'' error event
\begin{equation}
\calE_n^{\mathrm{noisy}} \!\!:=\! \{ |\calS^{\mathrm{noisy}}|\!>\!1\}\cup  ( \{ |\calS^{\mathrm{noisy}} |\!=\! 1 \} \cap \{ (\rvbX^*, \rvbw^*)\!\ne  \! (\bX, \bw) \} ). \nn
\end{equation}
Note that $(\calE_n^{\mathrm{noisy}})^c$ occurs, {\em both} the matrix $\bX$ and the noise vector $\bw$ are recovered so, in fact, we are decoding two objects when we are only interested in $\bX$.  Clearly, $\calE_n \subset \calE_n^{\mathrm{noisy}}$ so it suffices to upper bound $\bbP(\calE_n^{\mathrm{noisy}})$ to obtain an upper bound of $\bbP(\calE_n)$. For this purpose consider the event
\begin{equation}
\calA_{\bZ,\bv}^{\mathrm{noisy}} := \{ \lrangle{\bZ}{\rvbH_a} = \langle \bX,\rvbH_a\rangle +  v_a, \forall\, a\in [k] \},
\end{equation}
defined for each matrix-vector pair $(\bZ, \bv) \in\bbF_q^{n\times n}\times \bbF_q^k$ such that  $\rank(\bZ)+ \lambda \|\bv\|_0  \le   \rank(\bX) + \lambda \|\bw\|_0$. The error event $\calE_n^{\mathrm{noisy}}$ occurs if and only if there exists a pair $(\bZ,\bv)\ne (\bX, \bw)$ such that (i)  $\rank(\bZ)+ \lambda \|\bv\|_0  \le   \rank(\bX) + \lambda \|\bw\|_0$ and (ii) the event $\calA_{\bZ,\bv}^{\mathrm{noisy}}$ occurs. By the union of events bound, the error probability can be bounded as:
\begin{align}
\bbP(\calE_n^{\mathrm{noisy}}) 
&\le   \sum_{(\bZ, \bv): \rank(\bZ)+\lambda \|\bv\|_0 \le \rank(\bX)+\lambda \|\bw\|_0}  \bbP(\calA_{\bZ,\bv}^{\mathrm{noisy}}) \nn\\
&\eqa\sum_{(\bZ, \bv): \rank(\bZ)+\lambda \|\bv\|_0 \le \rank(\bX)+\lambda \|\bw\|_0}  q^{-k} \nn\\
&\leb q^{-k}|\calU_{r,s}| \label{eqn:qk},
\end{align}
where $(a)$  is from the same argument as the noiseless case [See~\eqref{eqn:probA}] and in $(b)$, we defined the set  $\calU_{r,s}:=\{ (\bZ, \bv): \rank(\bZ)+\lambda \|\bv\|_0 \le \rank(\bX)+\lambda \|\bw\|_0 \}$, where the subscripts $r$ and $s$ index respectively the upper bound on the rank of $\bX$ and sparsity of $\bw$.  Note that $s=\|\bw\|_0 = \floor{\sigma n^2} \le \sigma n^2$. It remains to bound the cardinality of $\calU_{r,s}$. In the following, we partition  the counting argument into disjoint subsets by fixing the sparsity of the vector $\bv$ to be equal to $l$ for all possible $l$'s. Note that $0\le l\le (\|\bv\|_0)_{\max} :=\frac{r}{\lambda}+s$.   The cardinality of $\calU_{r,s}$  is bounded as follows:
\begin{align}
| \calU_{r,s}|& = \sum_{l=0}^{(\|\bv\|_0)_{\max}} |\{ \bv\in \bbF_q^k: \|\bv\|_0=l\}| \times  \nn\\*
&\qquad\qquad\times |\{ \bZ\in\bbF_q^{n\times n} : \rank(\bZ) \le r+\lambda (s-l) \}| \nn\\
&\lea \sum_{l=0}^{(\|\bv\|_0)_{\max}} \left[ \binom{k}{l} (q-1)^l \right]  4q^{ 2n [r+\lambda (s-l)]-[r+\lambda (s-l)]^2 } \nn\\
&\leb \left( \frac{r}{\lambda}+s+1\right) \binom{k}{\frac{r}{\lambda}+s} q^{\frac{r}{\lambda}+s} 4 q^{2n(r+\lambda s)-(r+\lambda s)^2 } \nn\\
&\lec \left( \frac{r}{\lambda}+s+1\right) 2^{k H_2 ( \frac{\frac{r}{\lambda}+s}{k} )}  q^{\frac{r}{\lambda}+s}  4 q^{2n (r+\lambda s)-(r+\lambda s)^2 } , \nn
\end{align}
where $(a)$ follows by bounding the number of vectors which are non-zero in $l$ positions and the number of matrices whose rank is no greater than $r+\lambda (s-l)$ (Lemma~\ref{lem:num_mat}), $(b)$ follows by first noting that  the assignment  $r\mapsto 2nr-r^2$ is monotonically increasing in $r =0,1,\ldots, n$ and second by upper bounding the summands by their largest possible values.  Observe that~\eqref{eqn:knoisy} ensures  that $\frac{r}{\lambda}+s\le \frac{k}{2}$, which is needed to upper bound the binomial coefficient since $l\mapsto\binom{k}{l}$ is monotonically increasing iff $l\le k/2$. Inequality $(c)$ uses the fact that the binomial coefficient is upper bounded by a function of the binary entropy~\cite[Theorem~11.1.3]{Cov06}. Now, note that since $r/n\to \gamma$, for every $\eta>0$, $|r/n-\gamma|<\eta$ for $n$ sufficiently large. Define $\tilde{\gamma}_{\eta} := \gamma+\eta+\sigma$. From $(c)$ above,  $| \calU_{r,s}|$ can be further upper bounded as 
\begin{align}
&\led  4 \left(\tilde{\gamma}_{\eta}  n^2 +1\right) 2^{k H_2 ( \frac{\tilde{\gamma}_{\eta} n^2}{k} )}  q^{\tilde{\gamma}_{\eta} n^2}q^{2 \tilde{\gamma}_{\eta} n^2  - \tilde{\gamma}_{\eta}^2 n^2 } \label{eqn:up_bd2}\\
 & \lee O(n^2) 2^{k H_2(\frac{1}{  3 - \tilde{\gamma}_{\eta} }) }   q^{\tilde{\gamma}_{\eta} n^2+2\tilde{\gamma}_{\eta} n^2  - \tilde{\gamma}_{\eta}^2 n^2  }  
 \label{eqn:up_bd1}.
\end{align}
Inequality $(d)$ follows from the problem assumption that $\rank(\bX)\le r \le (\gamma+\eta)n$ for $n$ sufficiently large, $\|\bw\|_0=s\le\sigma n^2$ and the choice of the regularization  parameter $\lambda=1/n$.  Inequality $(e)$ follows from the fact that  since $k$ satisfies~\eqref{eqn:knoisy}, $k>3\tilde{\gamma}_{\eta} (1-\tilde{\gamma}_{\eta} /3)n^2$ and hence the binary entropy term in~\eqref{eqn:up_bd2} can be upper bounded as in \eqref{eqn:up_bd1}.  By combining \eqref{eqn:qk} and \eqref{eqn:up_bd1},   we observe that the error probability $\bbP(\calE_n^{\mathrm{noisy}})$ can be upper bounded as 
\begin{eqnarray}
\bbP(\calE_n^{\mathrm{noisy}})\le O(n^2) q^{-  n^2 \big[      \frac{k}{n^2}(1- (\log_q 2) H_2(\frac{1}{  3 - \tilde{\gamma}_{\eta} })    - 3\tilde{\gamma}_{\eta} +\tilde{\gamma}_{\eta}^2    \big] } \label{eqn:final_noisy}. 
\end{eqnarray}
Now, again by using  the assumption that   $k$ satisfies~\eqref{eqn:knoisy},   the exponent in~\eqref{eqn:final_noisy} is  positive for $\eta$ sufficiently small    ($\tilde{\gamma}_{\eta}\to \gamma+\sigma$ as $\eta\to 0$) and hence $\bbP(\calE_n^{\mathrm{noisy}})\to 0$ as $n\to\infty$.  
\end{proof}

%% file: noise_conv.tex
\section{Proof of Corollary~\ref{cor:conv_noise}} \label{app:noisy_conv}
\begin{proof}
Fano's inequality can be applied to obtain  inequality $(a)$ as in~\eqref{eqn:end_fanos}. We  lower bound the term  $H(\rvbX|\rvy^k,\rvbH^k)$ in~\eqref{eqn:end_fanos} differently taking into account the stochastic noise. It can be expressed as 
\begin{equation}
H(\rvbX|\rvy^k,\rvbH^k)=H(\rvbX)-H(\rvy^k|\rvbH^k)+ H(\rvy^k|\rvbH^k,\rvbX) . \label{eqn:entro_decomp}
\end{equation}
The second term can be upper bounded as $H(\rvy^k|\rvbH^k)\le k$ by~\eqref{eqn:equivo}. The third term, which is zero in the noiseless case, can be (more tightly) lower  bounded as follows:
\begin{eqnarray}
H(\rvy^k|\rvbH^k,\rvbX) = kH(\rvy_1|\rvbH_1,\rvbX) \eqa kH(\rvw_1)\geb k H_q(p), \label{eqn:ent_w}
\end{eqnarray}
where $(a)$ follows by the   independence of $(\rvbX,\rvbH_1)$ and $\rvw_1$ and   $(b)$ follows from the fact that the entropy of $\rvw$  with pmf in~\eqref{eqn:wdist} is lower bounded by putting all the remaining probability mass $p$ on a single symbol in $\bbF_q\setminus \{0\}$ (i.e., a $\mathrm{Bern}(p)$ distribution). Note that logarithms are to the base $q$. The result in~\eqref{eqn:conversek_noise} follows by uniting~\eqref{eqn:entro_decomp}, \eqref{eqn:ent_w} and the lower bound in~\eqref{eqn:less_r}. 
\end{proof}

%% file: noisy2_v2.tex
\section{Proof of Corollary~\ref{prop:noisy2} } \label{app:noisy2}
\begin{proof}
The main idea in the proof is to reduce the problem to the deterministic case and apply Proposition~\ref{prop:noisy}. For this purpose, we define the  {\em $\zeta$-typical set} (for the length-$k=\ceil{\alpha n^2}$ noise vector $\rvbw$) as
\begin{equation}
\calT_{\zeta}  =\calT_{\zeta}  (\rvw):= \left\{  \bw  \in\bbF_q^k:  \left|\,  \frac{\| \bw \|_0}{\alpha n^2} - p    \,  \right| \le \zeta       \right\}. \nn
\end{equation}
We choose $\zeta$ to be dependent on $n$  in the following way (cf.\ the Delta-convention~\cite{Csi97}): $\zeta_n\to 0$ and $n\zeta_n \to \infty$ (e.g., $\zeta_n=n^{-1/2}$). By Chebyshev's inequality,  $\bbP( \rvbw\notin  \calT_{\zeta_n} )\to 0$ as $n\to\infty$. We now bound the probability of error that the estimated matrix is not the same as the true one by using the law of total probability to condition the error event $\calE_n^{\mathrm{noisy}}$ on the event $\{\rvbw\in\calT_{\zeta_n} \}$ and its complement:
\begin{equation}
\bbP(\calE_n^{\mathrm{noisy}} )\le \bbP(\calE_n^{\mathrm{noisy}} |\rvbw\in\calT_{\zeta_n}) + \bbP(\rvbw\notin\calT_{\zeta_n}). \label{eqn:gf}
\end{equation}
Since the second term in \eqref{eqn:gf} converges to zero, it suffices to prove that the first term also converges to zero. For this purpose, we can follow the steps of the proof in Proposition~\ref{prop:noisy} and in particular the steps leading to~\eqref{eqn:up_bd2} and~\eqref{eqn:final_noisy}. Doing so and defining $p_{\zeta}:=p+\zeta$, we arrive at the upper bound
\begin{align}
&\bbP(\calE_n^{\mathrm{noisy}} |\rvbw\in  \calT_{\zeta_n} )  \nn\\
&  \le  \!  O(n^2) 2^{k H_2(\frac{\gamma n^2 + p_{\zeta_n} \alpha n^2 }{\alpha n^2} )} q^{(\frac{\gamma n^2 + p_{\zeta_n} \alpha n^2 }{\alpha n^2})}  \times  \nn\\
&\qquad\qquad \qquad \times q^{2n^2(\gamma  + p_{\zeta_n} \alpha ) - (\gamma n+p_{\zeta_n} \alpha n)^2-\alpha n^2 }\nn\\
&  \le  \!  O(n^2) q^{-n^2  \big[ \alpha -\alpha (\log_q 2) H_2(p_{\zeta_n}+\frac{\gamma}{\alpha} )   - 2\alpha p_{\zeta_n} (1-\gamma)  +\alpha^2 p_{\zeta_n}^2       -  2\gamma+\gamma^2  \big] }     \nn\\
& =\!  O(n^2)q^{-n^2 [ \, g(\alpha;p_{\zeta_n}, \gamma) - 2\gamma (1-\gamma/2) \,]}  , \label{eqn:noisy2_last}
\end{align}
Since     $\zeta_n\to 0$ and $g$   defined in~\eqref{eqn:defg}  is continuous in the second argument, $g(\alpha;p_{\zeta_n}, \gamma) 
\to  g(\alpha;p, \gamma)$. Thus, if $\alpha$  satisfies~\eqref{eqn:alpha_ineq}, the exponent  in~\eqref{eqn:noisy2_last} is positive. Hence,   $\bbP(\calE_n^{\mathrm{noisy}})\to 0$ as $n\to\infty$ as desired. 
\end{proof}

%% file: sparse_prf_v2.tex
\section{Proof  of Theorem~\ref{prop:sparse}}\label{app:sparse}
\begin{proof}
We first state a lemma which will be proven as the end of this section. 
\begin{lemma} \label{lem:circconv}
Define $d :=\|\bX  - \bZ\|_0$. The probability   of $\calA_{\bZ}$, defined in \eqref{eqn:AZ},  under the $\delta$-sparse measurement model, denoted as $\theta(d;\delta,q,k)$, is   a function of $d$ and is given as 
\begin{equation}
\theta(d; \delta, q,k):=\left[ q^{-1} + (1-q^{-1}) \left(1-\frac{\delta}{1-q^{-1}} \right)^d \right]^k . \label{eqn:p}
\end{equation}
\end{lemma}
Lemma~\ref{lem:circconv} says that the probability $\bbP(\calA_{\bZ})$ is {\em only} a function of $\bX$ though the number of entries it differs from $\bZ$, namely $d$. Furthermore, it is easy to check that the probability in~\eqref{eqn:p} satisfies the following two properties:
\begin{enumerate}
\item $\theta(d; \delta, q,k)\le (1-\delta)^k \le \exp(-k\delta)$ for all $d \in[n^2]$,
\item $\theta(d; \delta, q,k)$ is a monotonically decreasing function  in $d$. 
\end{enumerate}
We  upper bound  the probability in~\eqref{eqn:error_union}. To do so, we   partition all possibly {\em misleading} matrices $\bZ$ into subsets based on their Hamming distance from $\bX$. Our idea is to separately bound those partitions with low Hamming distance (which are few and so for which a loose upper bound on $\theta(d;\delta, q,k)$ suffices) and those further from $\bX$ (which are many, but for which we can get a tight upper bound on $\theta(d; \delta,q,k)$, a bound that is only a function of the Hamming distance $\ceil{\beta n^2}$).  Then we optimize the split over the free parameter $\beta$:
\begin{align}
&\bbP( \calE_n) \le \sum_{d=1}^{n^2} \sum_{\substack{\bZ:\bZ \ne \bX, \rank(\bZ)\le \rank(\bX)\\ \|\bX-\bZ\|_0=d }} \bbP(\calA_{\bZ})  \nn\\
&\eqa\sum_{d=1}^{\floor{\beta n^2}} \sum_{\substack{\bZ:\bZ \ne \bX, \rank(\bZ)\le \rank(\bX)\\ \|\bX-\bZ\|_0=d }} \theta(d;\delta,q,k)  + \nn\\*
&\qquad + \sum_{d=\ceil{\beta n^2}}^{n^2} \sum_{\substack{\bZ:\bZ \ne \bX, \rank(\bZ)\le \rank(\bX)\\ \|\bX-\bZ\|_0=d }} \theta(d;\delta,q,k) \nn\\
&\leb \sum_{d=1}^{\floor{\beta n^2}} \sum_{\substack{\bZ:\bZ \ne \bX, \rank(\bZ)\le \rank(\bX)\\ \|\bX-\bZ\|_0=d }} \exp(-k\delta)  + \nn\\*
&\qquad+ \sum_{d=\ceil{\beta n^2}}^{n^2} \sum_{\substack{\bZ:\bZ \ne \bX, \rank(\bZ)\le \rank(\bX)\\ \|\bX-\bZ\|_0=d }} \theta(\ceil{\beta n^2};\delta,q,k) \nn \\
&\lec  |\{ \bZ : \|\bZ-\bX\|_0\le\floor{\beta n^2}\}\exp(-k\delta) + \nn\\*
&\quad+ n^2 |\{ \bZ : \rank(\bZ)\le \rank(\bX)\} | \theta( \ceil{\beta n^2};\delta,q,k)  \label{eqn:upper_bd2} .
\end{align}
In $(a)$,  we used the definition of  $\theta(d; \delta, q,k)$ in Lemma~\ref{lem:circconv}. The  fractional parameter $\beta$, which we choose later, may depend on $n$.  In $(b)$, we used the fact that $\theta(d; \delta, q,k)\le \exp(-k\delta)$ and that $\theta(d; \delta, q,k)$ is monotonically decreasing in $d$ so $\theta(d;\delta, q, k)\le \theta(\ceil{\beta n^2};\delta,q,k)$ for all $d\ge \ceil{\beta n^2}$. In $(c)$, we  upper bounded the cardinality of the set $\{\bZ\ne \bX: \rank(\bZ)\le \rank(\bX), \|\bX-\bZ\|_0\le \floor{\beta n^2} \}$ by the  cardinality of the set of matrices that differ from $\bX$ in no more than  $\floor{\beta n^2}$ locations (neglecting the rank constraint). For the second term, we  upper bounded the cardinality of each set $\calM_d:=\{\bZ\ne \bX: \rank(\bZ)\le \rank(\bX), \|\bX-\bZ\|_0 = d\}$ by the cardinality of the  set of matrices whose rank no more than  $\rank(\bX)$ (neglecting the Hamming weight constraint). We denote the first and second terms in~\eqref{eqn:upper_bd2} as $A_n$ and $B_n$ respectively.   Now, 
\begin{align}
A_n &:=   |\{ \bZ : \|\bZ-\bX\|_0 \le \floor{\beta n^2}\} |\exp(-k\delta)  \nn\\
&\lea   2^{n^2 \Hb(\beta)}(q-1)^{\beta n^2} \exp(-k\delta) \nn\\
&\le    2^{n^2 \left[ \Hb(\beta)+\beta\log_2(q-1) - \frac{k}{n^2} \delta \log_2(e)\right] } , \label{eqn:trivialplus}
\end{align}
where $(a)$ used the fact that the number of matrices that differ from $\bX$ by  less than or equal to $\floor{\beta n^2}$ positions is upper bounded by $2^{n^2 \Hb(\beta)}(q-1)^{\beta n^2}$.  Note that this upper bound   is independent of $\bX$.     Now fix $\eta>0$ and consider $B_n$:
\begin{align}
&B_n := n^2|\{ \bZ : \rank(\bZ)\le \rank(\bX)\} |  \theta(\ceil{\beta n^2} ; \delta , q ,k) \nn\\
&\lea 4 n^2  q^{(2\gamma (1-\gamma/2) +\eta)n^2 } \theta(\ceil{\beta n^2} ; \delta , q ,k)  \nn \\
& \eqb  4 n^2 q^{n^2 \left[2\gamma(1-\gamma/2)+\eta+\frac{k}{n^2} \log_q \left(q^{-1} + (1-q^{-1})(1-\frac{\delta}{1-q^{-1}})^{\ceil{\beta n^2}}  \right) \right] } .  \label{eqn:trivialplus2}
\end{align}
In $(a)$, we used the fact that the number of matrices of rank no greater than $r$ is bounded above by $4 q^{(2\gamma (1-\gamma/2) +\eta) n^2  }$  (Lemma~\ref{lem:num_mat}) for $n$ sufficiently large (depending on $\eta$ by the convergence of $r/n$ to $\gamma$).  Equality $(b)$ is obtained by applying~\eqref{eqn:p} in  Lemma~\ref{lem:circconv}.


Our objective in the rest of the proof is to find sufficient conditions on $k$ and $\beta$ so that~\eqref{eqn:trivialplus} and~\eqref{eqn:trivialplus2} both converge to zero. We start with $B_n$. From~\eqref{eqn:trivialplus2} we observe that  if for every $\veps>0$, there exists an $N_{1,\veps}\in\bbN$ such that 
\begin{equation}
k> \left(1+\frac{\veps}{5} \right)  \frac{2\gamma (1-\gamma/2) n^2}{-\log_q  \left( q^{-1} + (1-q^{-1}) \left(1-\frac{\delta}{1-q^{-1}} \right)^{\ceil{\beta n^2}} \right) } , \label{eqn:k2}
\end{equation}
for all $n>N_{1,\veps}$,  then $B_n\to 0$ since the  exponent in~\eqref{eqn:trivialplus2} is negative (for $\eta$ sufficiently small). Now, we claim that if  $\lim_{n\to\infty}\ceil{\beta n^2} \delta=+\infty$ then the denominator in~\eqref{eqn:k2} tends to   $1$ from below. This is justified as follows: Consider the term,
\begin{align}
 \left(1-\frac{\delta}{1-q^{-1}} \right)^{\ceil{\beta n^2}} \le \exp\left( -\frac{  {\ceil{\beta n^2}}\delta}{1-q^{-1}} \right) \stackrel{ n\to\infty}{\longrightarrow} 0 ,\nn
\end{align}
so the argument of the logarithm in~\eqref{eqn:k2} tends to $q^{-1}$ from above if  $\lim_{n\to\infty}\ceil{\beta n^2} \delta=+\infty$.

Since $\delta\in \Omega(\frac{\log n}{n})$, by definition, there exists a constant $C \in (0,\infty)$ and an  integer $N_{\delta}\in\bbN$ such that
\begin{equation}
 \delta  =\delta_n\ge C \frac{\log_2(n)}{n}, \label{eqn:delta_choice}
\end{equation}
for all $n>N_{\delta}$. Let $\beta$ be defined as 
\begin{equation}
\beta= \beta_n :=  \frac{2  \gamma (1-\gamma/2)\log_2(e)  \delta}{\log_2(n) }  .\label{eqn:choice}
\end{equation} 
Then $\ceil{\beta n^2} \delta \ge 2  \gamma (1-\gamma/2)\log_2(e) C^2\log_2 (n) = \Theta(\log n)$ and so the condition  $\lim_{n\to\infty}\ceil{\beta n^2} \delta=+\infty$   is satisfied. Thus, for sufficiently large $n$, the denominator in \eqref{eqn:k2} exceeds $1/(1+\veps/5)<1$. As such, the condition in~\eqref{eqn:k2} can be equivalently written as: Given the choice of $\beta$ in~\eqref{eqn:choice}, if there exists an $N_{2,\veps}\in\bbN$ such that 
\begin{equation}
k> 2\left(1+\frac{\veps}{5} \right)^2  \gamma (1-\gamma/2) n^2   \label{eqn:k3}
\end{equation}
for all $n>N_{2,\veps}$, then $B_n\to 0$. 

We now revisit the upper bound on $A_n$ in~\eqref{eqn:trivialplus}. The inequality says that, for every $\veps>0$, if there exists an $N_{3,\veps}\in\bbN$ such that 
\begin{equation}
k > \left(1+\frac{\veps}{5} \right)\frac{ \Hb(\beta) +\beta \log_2(q-1)}{ \delta \log_2(e)} n^2, \label{eqn:k1}
\end{equation}
for all $n>N_{3,\veps}$, then $A_n\to 0$ since the exponent in~\eqref{eqn:trivialplus} is negative.  Note that $\Hb(\beta)/(-\beta\log_2\beta) \downarrow 1$ as $\beta \downarrow 0$. Hence, if $\beta$ is chosen as in \eqref{eqn:choice}, then by using~\eqref{eqn:delta_choice}, we  obtain
\begin{equation}
\lim_{n\to\infty}\frac{ \Hb(\beta) + \beta\log_2(q-1)}{ \delta \log_2(e)} \le 2\gamma(1-\gamma/2) . \label{eqn:seq_2}
\end{equation}
In particular, for $n$ sufficiently large, the terms in the    sequence in \eqref{eqn:seq_2} and its limit  (which exists)  differ by less than $2\gamma(1-\gamma/2)\veps/5$. Hence~\eqref{eqn:k1} is equivalent to the following:  Given the choice of $\beta$ in~\eqref{eqn:choice}, if  there exists an $N_{4,\veps}\in\bbN$ such that 
\begin{equation}
k> 2\left(1+\frac{\veps}{5} \right)^2 \gamma (1-\gamma/2) n^2  \label{eqn:k4}
\end{equation}
for all $n>N_{4,\veps}$, the sequence $A_n\to 0$.    The choice of $\beta$ in~\eqref{eqn:choice} ``balances''   the two  sums $A_n$ and $B_n$ in~\eqref{eqn:upper_bd2}. Also note that $2(1+\veps/5)^2<2+\veps $ for all $\veps\in (0,5/2)$. 

Hence,    if the number of measurements $k$ satisfies~\eqref{eqn:ach} for all $n>N_{\veps,\delta} :=\max\{N_{1,\veps}, N_{2,\veps}, N_{3,\veps}, N_{4,\veps}, N_{\delta}\}$, both~\eqref{eqn:k3} and~\eqref{eqn:k4} will also be satisfied and     consequently,    $\bbP(\calE_n) \le A_n+B_n\to 0$ as $n\to\infty$ as desired.  We remark that  the restriction of  $\veps\in (0,5/2)$  is not a serious one, since the validity of the claim in Theorem~\ref{prop:sparse}  for some $\veps_0>0$ implies the same for all $\veps>\veps_0$.  This completes the proof of Theorem~\ref{prop:sparse}.   \end{proof}



It now remains to prove Lemma~\ref{lem:circconv}.
\begin{proof}  
Recall that $d=\|\bX-\bZ\|_0$ and  $\theta(d;\delta, q, k)= \bbP(  \langle \rvbH_a, \bX \rangle= \langle \rvbH_a, \bZ \rangle, a\in [k] )$. By the i.i.d.\ nature of the random matrices $\rvbH_a, a\in [k]$, it is true that 
$$
\theta(d;\delta, q, k) = \bbP(  \langle \rvbH_1, \bX \rangle= \langle \rvbH_1, \bZ \rangle)^k. 
$$
It thus remains to demonstrate that 
\begin{equation}
\bbP(  \langle \rvbH_1, \bX \rangle= \langle \rvbH_1, \bZ \rangle)= q^{-1} + (1-q^{-1}) \left(1-\frac{\delta}{1-q^{-1}} \right)^d . \label{eqn:induct}
\end{equation}
This may be proved using induction on $d$ but we prove it using  more direct  transform-domain ideas. Note that  \eqref{eqn:induct} is simply the $d$-fold $q$-point circular convolution of the $\delta$-sparse pmf in~\eqref{eqn:prob_delta}. 
Let $\bF\in\bbC^{q\times q}$ and $\bF^{-1}\in\bbC^{q\times q}$ be the discrete Fourier transform (DFT) and the inverse DFT matrices respectively.  We use the convention in~\cite{Opp}. Let 
\begin{equation}
\bp := P_{\rvh}(\fndot; \delta,q) =  \begin{bmatrix} 1-\delta \\ \delta/(q-1)\\ \vdots\\\delta/(q-1)  \end{bmatrix}  \nn
\end{equation}
 be the vector of probabilities defined in~\eqref{eqn:prob_delta}. Then,  by properties of the DFT, \eqref{eqn:induct} is simply given by $\bF^{-1} [(\bF\bp)^{.d}]$ evaluated at the vector's first element. (The notation $\bv^{.d}:=[v_0^d \, \ldots \, v_{q-1}^d]^T$ denotes the vector in which each component of the vector $\bv$ is raised to the $d$-th power.)  We split $\bp$ into two vectors whose DFTs can be evaluated in closed-form:
\begin{equation}
\bp = \begin{bmatrix} \delta/(q-1) \\ \delta/(q-1)\\ \vdots\\ \delta/(q-1) \end{bmatrix}  + \begin{bmatrix} 1-\delta-\delta/(q-1) \\ 0\\ \vdots\\ 0 \end{bmatrix}  .  \nn
\end{equation}
Let the first and second vectors above be $\bp_1$ and $\bp_2$ respectively. Then,  by linearity of the DFT, $\bF\bp  = \bF\bp_1 +\bF\bp_2$ where 
\begin{equation}
\bF\bp_1 = \begin{bmatrix}  q\delta/(q-1) \\ 0\\ \vdots\\ 0 \end{bmatrix} ,\quad \bF\bp_2 = \begin{bmatrix}  1-\delta-\delta/(q-1) \\  1-\delta-\delta/(q-1) \\ \vdots\\ 1-\delta-\delta/(q-1) \end{bmatrix}  . \nn
\end{equation}
Summing these up yields
\begin{equation}
\bF\bp =  \begin{bmatrix}  1 \\ 1-  \delta/(1-q^{-1}) \\ \vdots\\ 1-  \delta/(1-q^{-1}) \end{bmatrix} . \nn
\end{equation}
Raising $\bF\bp$ to the $d$-th power yields
\begin{equation}
(\bF\bp)^{.d}= \begin{bmatrix}  1 \\ (1-  \delta/(1-q^{-1}))^d \\ \vdots\\ (1-  \delta/(1-q^{-1}))^d \end{bmatrix} . \nn
\end{equation}
Now using the same splitting technique, $(\bF\bp)^{.d}$ can be decomposed into  
\begin{equation}
(\bF\bp)^{.d} \!=\!  \begin{bmatrix}  (1-  \delta/(1-q^{-1}))^d  \\ (1-  \delta/(1-q^{-1}))^d \\ \vdots\\ (1-  \delta/(1-q^{-1}))^d \end{bmatrix} +\begin{bmatrix}  1\!-\!(1-  \delta/(1-q^{-1}))^d  \\ 0 \\ \vdots\\ 0 \end{bmatrix} . \nn
\end{equation}
Let $\bs_1$ and $\bs_2$ denote each vector  on the right hand side above. Define $\varphi:= (1-  \delta/(1-q^{-1}))^d$. Then, the inverse DFTs of $\bs_1$ and $\bs_2$ can be evaluated analytically as
\begin{align}
 \bF^{-1}\bs_1 = \begin{bmatrix}  \varphi  \\ 0 \\ \vdots\\0\end{bmatrix} ,  \qquad\bF^{-1}\bs_2 =  \begin{bmatrix}  q^{-1}(   1-\varphi) \\   q^{-1}(   1-\varphi)  \\ \vdots\\  q^{-1}(   1-\varphi)  \end{bmatrix} .  \nn
\end{align}
Summing the first elements of $\bF^{-1}\bs_1$ and $\bF^{-1} \bs_2$  completes the proof of~\eqref{eqn:induct}  and hence   of Lemma~\ref{lem:circconv}. 
\end{proof}

%% file: stats_v2.tex
\section{Proof of Lemma~\ref{lem:stats}}\label{app:stats}
\begin{proof}
The only matrix for which the rank $r=0$ is the zero matrix which is in $\scC$, since $\scC$ is a linear code (i.e., a subspace).  Hence, the sum in~\eqref{eqn:NCr_def} consists only of a single term, which is one.   Now for $1\le r \le n$, we start from~\eqref{eqn:NCr_def} and by the linearity of expectation, we have
\begin{align}
\bbE  \rvN_{\scC}(r) &=  \sum_{\bM \in \bbF_q^{n\times n}: \rank(\bM)=r} \bbE \,\bbI\{\bM\in\scC\} \nn  \\
&=  \sum_{\bM \in \bbF_q^{n\times n}: \rank(\bM)=r} \bbP( \bM\in\scC ) \nn\\
&\eqa \sum_{\bM \in \bbF_q^{n\times n}: \rank(\bM)=r} q^{-k}=\Phi_q(n,r) q^{-k}, \nn
\end{align}
where $(a)$ is because $\bM\ne \bzero$ (since $1\le r \le n$). Hence, as in~\eqref{eqn:probA}, $\bbP( \bM\in\scC )=q^{-k}$. The proof is completed by appealing to~\eqref{eqn:exact_r}, which provides upper and lower bounds on the number of matrices of rank exactly $r$. For the variance,  note that the random variables in the set $\{\bbI\{\bM\in\scC\}: \rank(\bM)=r\}$ are pairwise independent (See Lemma~\ref{lem:pairwise}). As a result, the variance of the sum in~\eqref{eqn:NCr_def} is a sum of variances, i.e., 
\begin{align}
\var&(\rvN_{\scC}(r)) = \sum_{\bM \in \bbF_q^{n\times n}: \rank(\bM)=r} \var(\bbI\{\bM\in\scC\} )\nn\\
&= \sum_{\bM \in \bbF_q^{n\times n}: \rank(\bM)=r} \bbE\left[\bbI\{\bM\in\scC\}^2 \right] - [\bbE\,\bbI\{\bM\in\scC\}]^2\nn\\
&\le \sum_{\bM \in \bbF_q^{n\times n}: \rank(\bM)=r} \bbE\,\bbI\{\bM\in\scC\}  = \bbE \rvN_{\scC}(r),\nn
\end{align}
as desired. 
\end{proof}

%% file: li_prf_v4.tex
\section{Proof of Proposition~\ref{prop:li}} \label{app:li}
\begin{proof}
We first restate a beautiful result  from~\cite{Blomer}. For each positive integer $k$, define the interval $\calI_k:=[ \frac{\log_e k}{k}, \frac{q-1}{q} ]$. 
\begin{theorem}[Corollary 2.4 in \cite{Blomer}]\label{thm:blom}
Let $\rvbM$ be a random $k\times k$ matrix over the finite field $\bbF_q$, where each element is drawn independently from the pmf in~\eqref{eqn:prob_delta} with $\delta$, a sequence in $k$, belonging to $\calI_k$ for each $k\in\bbN$. Then, for every   $l\le k$,
\begin{equation}
\bbP( k-\rank(\rvbM) \ge l)\le Aq^{-l},  \label{eqn:blomer}
\end{equation}
and $A$ is a constant. Moreover, if  $A$ is considered as a function of $\delta$ then it is monotonically decreasing  as a function in  the interval  $\calI_k$.
\end{theorem}
To prove the Proposition~\ref{prop:li}, first define $N:=n^2$ and let $\rvbh_a:=\vect(\rvbH_a) \in \bbF_q^N$ be the vectorized versions of the random sensing matrices. Also let $\rvbH := [\rvbh_1 \,\ldots \,  \rvbh_k]\in \bbF_q^{N\times k}$ be the matrix with columns $\rvbh_a$. Finally, let $\rvbH_{[k\times k]} \in \bbF_q^{k\times k}$ be the square sub-matrix of $\rvbH$ consisting only of its top $k$ rows. Clearly, the dimension of the column span of $\rvbH$, denoted as  $\rvm \ge\rank(\rvbH_{[k\times k]})$. Note that $\rvm$ is a sequence of random variables and $k$ is a sequence of integers but we suppress their dependences on $n$.  Fix $0<\epsilon<1$ and consider  
\begin{align}
\bbP\left( \left|\frac{\rvm}{k}-  1\right|\ge \epsilon \right) &= \bbP\left(  \frac{\rvm}{k}\le 1- \epsilon \right) \nn\\
&\le\bbP\left(\frac{\rank(\rvbH_{[k\times k]} )}{k} \le 1-\epsilon \right) \nn\\
&=\bbP\left(k- \rank(\rvbH_{[k\times k]})  \ge  \epsilon k\right)   \nn\\
&\lea Aq^{-\epsilon k}, \label{eqn:borel_cant}
\end{align}
where for $(a)$   recall that $k\in\Theta(n^2)$  and    $\delta\in\Omega(\frac{\log n}{n})$. These facts imply that $\delta$ (as a sequence in $n$) belongs to  the interval $\calI_{k}$ for all    sufficiently large $n$ [because any function in $\Omega(\frac{\log n}{n})$ dominates the lower bound $\frac{\log_e k}{k}$ for $k\in\Theta(n^2)$] so the hypothesis of Theorem~\ref{thm:blom} is satisfied and we can apply~\eqref{eqn:blomer} (with $l=\epsilon k$) to get inequality $(a)$.    Since~\eqref{eqn:borel_cant} is a summable sequence, by the Borel-Cantelli lemma, the sequence of random variables $\rvm/{k}\to 1$ a.s. 
\end{proof}

%% file: rank_min_v2.bbl
\begin{thebibliography}{10}

\bibitem{Tan11}
V.~Y.~F. Tan, L.~Balzano, and S.~C. Draper, ``Rank minimization over finite
  fields,'' in {\em Intl.\ Symp.\ Inf.\ Th.}, (St Petersburg, Russia), Aug
  2011.

\bibitem{Can10}
E.~J. Cand\`{e}s and T.~Tao, ``{The power of convex relaxation: near-optimal
  matrix completion},'' {\em IEEE Trans.\ on Inf.\ Th.}, vol.~56,
  pp.~2053--2080, May 2010.

\bibitem{Can09}
E.~J. Cand\`{e}s and B.~Recht, ``{Exact matrix completion via convex
  optimization},'' {\em Foundations of Computational Mathematics}, vol.~9,
  no.~6, pp.~717--772, 2009.

\bibitem{Rec09}
B.~Recht, ``{A simpler approach to matrix completion},'' {\em To appear in J.
  Mach. Learn. Research}, 2009.
\newblock arXiv:0910.0651v2.

\bibitem{Rec09a}
B.~Recht, M.~Fazel, and P.~A. Parrilo, ``{Guaranteed minimum-rank solutions of
  linear matrix equations via nuclear norm minimization},'' {\em SIAM Rev.},
  vol.~2, no.~52, pp.~471--501, 2009.

\bibitem{Meka}
R.~Meka, P.~Jain, and I.~S. Dhillon, ``{Guaranteed rank minimization via
  singular value projection},'' in {\em Proc. of Neural Information Processing
  Systems}, 2010.
\newblock arXiv:0909.5457.

\bibitem{Gab85}
E.~M. Gabidulin, ``{Theory of codes with maximum rank distance},'' {\em Probl.
  Inform. Transm.}, vol.~21, no.~1, pp.~1--12, 1985.

\bibitem{Roth}
R.~M. Roth, ``Maximum-rank array codes and their application to crisscross
  error correction,'' {\em IEEE Trans. on Inf. Th.}, vol.~37, pp.~328--336, Feb
  1991.

\bibitem{Loi06}
P.~Loidreau, ``{Properties of codes in rank metric},'' 2006.
\newblock arXiv:0610057.

\bibitem{Sil08}
D.~Silva, F.~R. Kschischang, and R.~K\"{o}tter, ``{A rank-metric approach to
  error control in random network coding},'' {\em IEEE Trans.\ on Inf.\ Th.},
  vol.~54, pp.~3951 -- 3967, Sep 2008.

\bibitem{Mon07}
A.~Montanari and R.~Urbanke, ``{Coding for network coding},'' 2007.
\newblock arXiv:0711.3935.

\bibitem{Gad08}
M.~Gadouleau and Z.~Yan, ``{Packing and covering properties of rank metric
  codes},'' {\em IEEE Trans. on Inf. Th.}, vol.~54, pp.~3873--3883, Sep 2008.

\bibitem{NetflixPrize}
ACM SIGKDD and Netflix, {\em Proceedings of KDD Cup and Workshop}, (San Jose,
  CA), Aug 2007.
\newblock Proceedings available online at
  \url{http://www.cs.uic.edu/~liub/KDD-cup-2007/proceedings.html}.

\bibitem{Faz01}
M.~Fazel, H.~Hindi, and S.~P. Boyd, ``{A Rank Minimization Heuristic with
  Application to Minimum Order System Approximation},'' in {\em American
  Control Conference}, 2001.

\bibitem{Faz03}
M.~Fazel, H.~Hindi, and S.~P. Boyd, ``{Log-det heuristic for matrix rank
  minimization with applications with applications to Hankel and Euclidean
  distance metrics},'' in {\em American Control Conference}, 2003.

\bibitem{Bar11}
Z.~Bar-Yossef, Y.~Birk, T.~S. Jayram, and T.~Kol, ``{Index coding with side
  information},'' {\em IEEE Trans.\ on Inf.\ Th.}, vol.~57, pp.~1479 -- 1494,
  Mar 2011.

\bibitem{Roth97}
R.~M. Roth, ``Probabilistic crisscross error correction,'' {\em IEEE Trans. on
  Inf. Th.}, vol.~43, pp.~1425--1438, May 1997.

\bibitem{Sil10}
D.~Silva, F.~R. Kschischang, and R.~K\"{o}tter, ``{Communication over
  finite-field matrix channels},'' {\em IEEE Trans.\ on Inf.\ Th.}, vol.~56,
  pp.~1296 -- 1305, Mar 2010.

\bibitem{Barg02}
A.~Barg and G.~D. Forney, ``{Random codes: Minimum distances and error
  exponents},'' {\em IEEE Trans. on Inf. Th.}, vol.~48, pp.~2568--2573, Sep
  2002.

\bibitem{Gall}
R.~G. Gallager, {\em Low density parity check codes}.
\newblock MIT Press, 1963.

\bibitem{Kot08}
{R. K\"{o}tter and F. R. Kschischang}, ``Coding for errors and erasures in
  random network coding,'' {\em IEEE Trans.\ on Inf.\ Th.}, vol.~54, pp.~3579
  -- 3591, Aug 2008.

\bibitem{Nob11}
R.~W. N\'{o}brega, B.~F. Uch\^{o}a-Filho, and D.~Silva, ``On the capacity of
  multiplicative finite-field matrix channels,'' in {\em Intl.\ Symp.\ Inf.\
  Th.}, (St Petersburg, Russia), Aug 2011.

\bibitem{deCaen}
D.~de~Caen, ``{A lower bound on the probability of a union},'' {\em Discrete
  Math.}, vol.~69, pp.~217--220, May 1997.

\bibitem{Seg98}
{G. E. S\'{e}guin}, ``{A lower bound on the error probability for signals in
  white Gaussian noise},'' {\em IEEE Trans.\ on Inf.\ Th.}, vol.~44,
  pp.~3168–--3175, Jul 1998.

\bibitem{Coh04}
A.~Cohen and N.~Merhav, ``{Lower bounds on the error probability of block codes
  based on improvements on de Caen's inequality},'' {\em IEEE Trans. on Inf.
  Th.}, vol.~50, pp.~290--310, Feb 2004.

\bibitem{Bar10}
D.~Baron, S.~Sarvotham, and R.~G. Baraniuk, ``Bayesian compressive sensing via
  belief propagation,'' {\em IEEE Trans. on Sig. Proc.}, vol.~51, pp.~269 --
  280, Jan 2010.

\bibitem{Eldar11}
Y.~C. Eldar, D.~Needell, and Y.~Plan, ``{Unicity conditions for low-rank matrix
  recovery},'' {\em Preprint}, Apr 2011.
\newblock arXiv:1103.5479 (Submitted to SIAM Journal on Mathematical Analysis).

\bibitem{Pap10}
D.~S. Papailiopoulos and A.~G. Dimakis, ``{Distributed storage codes meet
  multiple-access wiretap channels},'' in {\em Proc. of Allerton}, 2010.

\bibitem{Dra09}
S.~C. Draper and S.~Malekpour, ``Compressed sensing over finite fields,'' in
  {\em Intl.\ Symp.\ Inf.\ Th.}, (Seoul, Korea), July 2009.

\bibitem{Vis10}
S.~Vishwanath, ``{Information theoretic bounds for low-rank matrix
  completion},'' in {\em Intl.\ Symp.\ Inf.\ Th.}, (Austin, TX), July 2010.

\bibitem{Emad11}
A.~Emad and O.~Milenkovic, ``Information theoretic bounds for tensor rank
  minimization,'' in {\em Proc. of Globecomm}, Dec 2011.
\newblock arXiv:1103.4435.

\bibitem{Kak11}
A.~Kakhaki, H.~K. Abadi, P.~Pad, H.~Saeedi, K.~Alishahi, and F.~Marvasti,
  ``{Capacity achieving random sparse linear codes},'' {\em Preprint}, Aug
  2011.
\newblock arXiv:1102.4099v3.

\bibitem{Gro81}
M.~Gr\"{o}tchel, L.~Lov\'{a}sz, and A.~Schrijver, ``{The ellipsoid method and
  its consequences in combinatorial optimization},'' {\em Combinatorica},
  vol.~1, no.~2, pp.~169--197, 1981.

\bibitem{Cor03}
T.~Cormen, C.~Leiserson, R.~Rivest, and C.~Stein, {\em Introduction to
  Algorithms}.
\newblock McGraw-Hill Science/Engineering/Math, 2nd~ed., 2003.

\bibitem{Cov06}
T.~M. Cover and J.~A. Thomas, {\em Elements of Information Theory}.
\newblock Wiley-Interscience, 2nd~ed., 2006.

\bibitem{Csi82}
I.~Csisz\'{a}r, ``{Linear codes for sources and source networks: Error
  exponents, universal coding},'' {\em IEEE Trans.\ on Inf.\ Th.}, vol.~28,
  pp.~585--592, Apr 1982.

\bibitem{Gal}
R.~G. Gallager, {\em Information Theory and Reliable Communication}.
\newblock Wiley, 1968.

\bibitem{SilPersonal}
D.~Silva. Personal communication, Sep 2011.

\bibitem{Ksc01}
F.~Kschischang, B.~Frey, and H.-A. Loeliger, ``Factor graphs and the
  sum-product algorithm,'' {\em IEEE Trans. on Inf. Th.}, vol.~47,
  pp.~498--519, Feb 2001.

\bibitem{Dembo}
A.~Dembo and O.~Zeitouni, {\em Large Deviations Techniques and Applications}.
\newblock Springer, 2nd~ed., 1998.

\bibitem{Lidl}
R.~Lidl and H.~Niederreiter, {\em Introduction to Finite Fields and their
  Applications}.
\newblock Cambridge University Press, 1994.

\bibitem{Blomer}
{J. Bl\"{o}mer, R. Karp and E. Welzl}, ``{The Rank of Sparse Random Matrices
  over Finite Fields},'' {\em Random Structures and Algorithms}, vol.~10,
  no.~4, pp.~407--419, 1997.

\bibitem{Cha96}
F.~Chabaud and J.~Stern, ``The cryptographic security of the syndrome decoding
  problem for rank distance codes,'' in {\em ASIACRYPT}, pp.~368--381, 1996.

\bibitem{Our02}
A.~V. Ourivski and T.~Johansson, ``New technique for decoding codes in the rank
  metric and its cryptography applications,'' {\em Probl. Inf. Transm.},
  vol.~38, pp.~237--246, July 2002.

\bibitem{Rich04}
G.~Richter and S.~Plass, ``{Error and erasure of rank-codes with a modified
  Berlekamp-Massey algorithm},'' in {\em Proceedings of ITG Conference on
  Source and Channel Coding}, Jan 2004.

\bibitem{Peeters}
R.~Peeters, ``{Orthogonal representations over finite fields and the chromatic
  number of graphs},'' {\em Combinatorica}, vol.~16, no.~3, pp.~417--431, 1996.

\bibitem{Lov81}
{L. Lov\'{a}sz}, ``{On the Shannon capacity of a graph},'' {\em IEEE Trans.\ on
  Inf.\ Th.}, vol.~IT-25, pp.~1--7, Jan 1981.

\bibitem{Dimakis09}
A.~G. Dimakis and P.~O. Vontobel, ``{LP Decoding meets LP Decoding: A
  Connection between Channel Coding and Compressed Sensing},'' in {\em Proc.\
  of Allerton}, 2009.

\bibitem{Csi97}
I.~Csisz\'ar and J.~Korner, {\em Information Theory: Coding Theorems for
  Discrete Memoryless Systems}.
\newblock Akademiai Kiado, 1997.

\bibitem{Opp}
A.~V. Oppenheim, R.~W. Schafer, and J.~R. Buck, {\em Discrete-Time Signal
  Processing}.
\newblock Prentice Hall, 1999.

\end{thebibliography}
